\documentclass{amsart}                    

\usepackage{graphicx}
\numberwithin{equation}{section}

\newtheorem{theorem}{Theorem}[section]
\newtheorem{corollary}[theorem]{Corollary}
\newtheorem{lemma}[theorem]{Lemma}

\theoremstyle{definition}

\newtheorem{remark}[theorem]{Remark}

\newcommand{\rset}{\mathbf{R}}
\newcommand{\tset}{\mathbf{T}}
\begin{document}

\begin{abstract}
Ehrenfest, Born-Oppenheimer, Langevin and Smoluchowski 
dynamics 
are shown to be  accurate approximations of time-independent Schr\"odinger observables 
for a molecular system avoiding caustics,
in the limit of large ratio of nuclei and electron masses, without assuming that
the nuclei are localized to vanishing domains. 
The derivation,
based on a Hamiltonian system interpretation of the Schr\"odinger equation 
and stability of the corresponding Hamilton-Jacobi equation,
bypasses the usual separation
of nuclei and electron wave functions, includes crossing electron eigenvalues, 
and gives a different perspective on 
the Born-Oppenheimer approximation, Schr\"odinger Hamiltonian systems, 
stochastic electron equilibrium states and numerical simulation
in molecular dynamics modeling.
\end{abstract}

\title[Molecular Dynamics Derived from the Schr\"odinger Equation]
{Stochastic and Deterministic Molecular Dynamics Derived from the
Time-independent Schr\"odinger Equation }


\author{Anders Szepessy}
\address{Department of Mathematics\\
  Kungl. Tekniska H\"ogskolan\\
  100 44 Stockholm\\
  Sweden}
\email{szepessy@kth.se}

\maketitle

\tableofcontents

\section{The Schr\"odinger  and molecular dynamics models}

The {\it time-independent  Schr\"odinger} equation 
\begin{equation}\label{schrodinger_stat}
H(x,X) \Phi(x,X) = E\Phi(x,X),
\end{equation}
models nuclei-electron systems and  is
obtained from minimization of the energy in the solution space of wave functions,  cf. 
\cite{schrodinger,schiff,berezin,tanner,lebris_hand}.
It is an eigenvalue problem for the energy  $E\in\rset$ of the system in the solution space, 
described by wave functions, $\Phi:\rset^{3J}\times\rset^{3N}\rightarrow \mathbb{C}$, depending on electron coordinates $x=(x^1,\ldots, x^J)\in\rset^{3J}$,
nuclei coordinates $X=(X^1,\ldots, X^N)\in\rset^{3N}$,  and a 
Hamiltonian operator $H(x,X)$ 
\begin{equation}\label{V-definition}
H(x,X)= V(x,X) - \frac{1}{2} M^{-1}\sum_{n=1}^N\Delta_{X^n}.
\end{equation}
The nuclei masses  $M$ are assumed to be large  and 
the interaction potential $V$, independent of $M$, is in the canonical
setting (neglecting relativistic and magnetic effect),
\begin{equation} \label{ham}
\begin{array}{ll}
V(x,X)  =&  - \frac{1}{2}\sum_{j=1}^J\Delta_{x^j} + \sum_{1\le k < j\le J}\frac{1}{|x^k-x^j|}  \\
&- \sum_{n=1}^N \sum_{j=1}^J \frac{Z_n}{|x^j-X^n|}+ \sum_{1\le n<m\le N}\frac{Z_nZ_m}{|X^n-X^m|}, 
\end{array}
\end{equation}
composed of the kinetic energy of the electrons, the electron-electron repulsion,
the electron-nuclei attraction, and the repulsion of nuclei (with charge $Z_n$), in the Hartree atomic units
where the electron mass, electron charge, reduced Planck constant, and the Coulomb force constant 
$(4\pi\epsilon_0)^{-1}$ all are one.
The mass of the nuclei, which are much greater than one (electron mass),  
are the diagonal elements in the diagonal matrix $M$. 

An essential feature of  the partial differential equation (\ref{schrodinger_stat})
is the high computational complexity to determine  the solution, in an antisymmetric/symmetric
subset of the 
Sobolev space $H^1(\rset^{3(J+N)})$.
Wave functions depend also on discrete spin states
\begin{equation}\label{spin_def}
\Phi(x^1,\sigma_1, \ldots,  x^J, \sigma_J, X_1,\Sigma_1, \ldots,X_N,\Sigma_N),
\end{equation}
which effect the solutions space:
each electron spin $\sigma_j$ can take two different values and
each nucleus can be in a finite set of spin states  $\Sigma_n$;
the {\it Pauli exclusion principle} restricts the solutions space to
wave functions satisfying the antisymmetry/symmetry
\[
\Phi(\ldots,x^j,\sigma_j,\ldots,x^k,\sigma_k,\ldots )=-\Phi(\ldots,x^k,\sigma_k,\ldots,x^j,\sigma_j,\ldots ) 
\mbox{  for any $1\le j,k\le J$}\]
and  for any pair of nuclei $n$ and $m$, with $A$ nucleons and the same number of protons and neutrons,
\[
\Phi(\ldots,X^m,\Sigma_m,\ldots,X^n,\sigma_n,\ldots )=(-1)^{A}\Phi(\ldots,X^n,\Sigma_n,\ldots,X^m,\Sigma_m,\ldots ) 
\] 
cf. \cite{lebris_hand}. We simplify the notation  by writing $\Phi(x,X)$
 instead of the more complete (\ref{spin_def}),
 since the Hamiltonian $H$ does not depend on the spin. % of each particle.
 The time-independent Schr\"odinger equation has convincing agreement
with experimental results,  as the basis for computational chemistry
and solid state physics.
An attractive property of the Schr\"odinger equation (\ref{schrodinger_stat})
is the precise definition of the Hamiltonian and the solutions space,
without unknown parameters. 
The agreement
with measurements can be further
improved by including relativistic and magnetic effects, cf. \cite{lebris_hand}.

In contrast to the Schr\"odinger equation, 
a {\it molecular dynamics} model of nuclei $X:[0,T]\rightarrow\rset^{3N}$, with a given potential $V_p:\rset^{3N}\rightarrow \rset$,
can be computationally studied also for large $N$  by solving the ordinary differential equation 
\begin{equation}\label{md_eq}
M\ddot X_\tau=- \partial_XV_p(X_\tau).
\end{equation}
This computational and conceptual simplification motivates the study to determine 
the potential and its implied accuracy by  a derivation
of molecular dynamics
from the Schr\"odinger equation,  as started  already
in the 1920's
with the seminal Born-Oppenheimer approximation \cite{BO}. 
The purpose here is to contribute to the current understanding
of such derivations, by showing %improved 
convergence rates under new assumptions.
The precise aim in this paper is to estimate the error 
\begin{equation}\label{approximation}
\int_{\rset^{3N+3J}} g(X) \Phi(x,X)^*\Phi(x,X) dx dX - \lim_{T\rightarrow\infty}T^{-1}\int_0^T g(X_\tau) d\tau
\end{equation}
of a given
position observable $\int g(X) \Phi(x,X)^*\Phi(x,X) dx dX$ 
of the time-indepedent Schr\"odinger equation \eqref{schrodinger_stat}
approximated 
by the corresponding molecular dynamics observable $\lim_{T\rightarrow\infty}T^{-1}\int_0^T g(X_\tau) d\tau$,
which is computationally cheaper to evaluate for several nuclei.

 A useful sub step to derive molecular dynamics  from the Schr\"odinger equation is
{\it Ehrenfest dynamics},  for  classical {\it ab initio} motion of the nuclei
coupled to Schr\"odinger dynamics for the electrons,
\begin{equation}\label{psi_X_eq}
\begin{array}{ll}
M\ddot {X}^n_\tau&= -\int_{\rset^{3J}} 
\phi_\tau^*(x, X_\tau) \,  {\partial_{X^n}V(x,X_\tau)}\, \phi_\tau(x, X_\tau)\ dx \\
i\dot \phi_\tau &=V(\cdot, X_\tau)\phi_\tau, 
\end{array}
\end{equation}
with the initial normalization $\int_{\rset^{3J}}\phi_0^*(x, X_0)\phi_0(x, X_0) dx=1$.
The Ehrenfest dynamics
(\ref{psi_X_eq})  has been derived from the time-dependent 
Schr\"odinger equation through the self consistent field equations, see   \cite{schutte,marx,tully}.
Equation (\ref{psi_X_eq}) can be used for  {\it ab initio}  computation of  molecular dynamics, cf.  \cite{marx,tully2}.
A next step  is the zero-order Born-Oppenheimer approximation, where $X_\tau$ solves  the classical
{\it ab initio} molecular dynamics (\ref{md_eq}) with
the potential $V_p:\rset^{3N}\rightarrow \rset$ determined as an eigenvalue of the electron Hamiltonian
$V(\cdot,X)$ for a given nuclei position $X$, that is  $V_p=\lambda_0$
and $V(\cdot,X)\psi_0(\cdot, X)=\lambda_0(X)\psi_0(\cdot, X)$
for an  electron eigenfunction $\psi_0(\cdot, X)\in L^2(\rset^{3J})$, for instance  the ground state.

The Born-Oppenheimer expansion \cite{BO} is an approximation of the solution to the time-independent\break
Schr\"odinger equation
which is shown \cite{hagedorn_egen,martinez} 
to solve the time-independent Schr\"odinger equation approximately. %with arbitrary high accuracy.
This expansion, analyzed by the methods of multiple scales, pseudo differential operators
and spectral analysis in \cite{hagedorn_egen,martinez, fefferman},
can be used to study the approximation error
\eqref{approximation}, cf. Section \ref{wkb_utan_elektroner}. Although in the literature it seems
easier to find precise statements      
on the error in observables for the other setting of the time-dependent Schr\"odinger equation.  
Instead of an asymptotic expansion,  we use a different method based on 
a Hamiltonian dynamics formulation of the time-independent
Schr\"odinger eigenfunction 
and  the stability of the corresponding perturbed
Hamilton-Jacobi equations viewed as a hitting problem.
% which gives convergence rates under different assumptions
%than the methods based on an expansion.
Another motivation for our method  is 
that it forms a sub step in trying to  estimate the approximation error
using only information available in molecular dynamics simulations.
%The question focused on here - how accurate does 
%the Born-Oppenheimer  dynamics approximate
%observables to the time-indepedent Schr\"odinger equation? - seems less studied,
%with the exceptions of the one dimensional case, cf. \cite{fefferman}.

The related problem of approximating 
observables to the time-dependent Schr\"odinger by the Born Oppenheimer expansions
is well studied, theoretically \cite{robert, spohn_egorov} and computationally \cite{lasser}
using the Egorov theorem.
The Egorov theorem shows that finite time
observables of the time-dependent Schr\"odinger equation are approximated with $\mathcal{O}(M^{-1})$
accuracy by the zero-order Born-Oppenheimer dynamics, with an electron eigenvalue gap.
In the special case of a position observable and no electrons (i.e. V=V(X) in \eqref{V-definition}), 
the Egorov theorem states that
\begin{equation}\label{egorov}
\Big|\int_{\rset^{3N}} g(X) \Phi(X,t)^*\Phi(X,t) \, dX-
\int_{\rset^{3N}} g(X_t) \Phi(X_0,0)^* \Phi(X_0,0) \, dX_0\Big|
\le C_t M^{-1},
\end{equation}
where $\Phi(X,t)$  is a solution to the time-dependent
Schr\"odinger equation $iM^{-1/2}\partial_t \Phi(\cdot,t)=H\Phi(\cdot,t)$ with the Hamiltonian \eqref{V-definition}
and the path $X_t$ is the nuclei coordinates for the dynamics with the Hamiltonian $|\dot X|^2/2 + V(X)$.
If the initial wave function $\Phi(X,0)$ is the eigenfunction in \eqref{schrodinger_stat}
the first term in \eqref{egorov} reduces to the first term in \eqref{approximation}
and the second terms can also become the same in an ergodic limit; however
since we do not know that the parameter $C_t$ (bounding an integral over $(0,t)$) is bounded for all time
we cannot directly conclude an estimate for \eqref{approximation} from \eqref{egorov}.
%
% also in the case of caustics) 
In this perspective, to study the time-independent instead of the time-dependent Schr\"odinger equation has
the important differences that
\begin{itemize}
\item the infinite time study of  the Born-Oppenheimer dynamics can be reduced
to a finite time hitting problem, 
\item the computational and theoretical problem of specifying initial data for the Schr\"odinger equation is avoided, and
\item computational cheap evaluation of a position observable $g(X)$ 
is possible using the time average $\lim_{T\rightarrow\infty}\int_0^T g(X_\tau)d\tau/T$ along the solution path $X_\tau$.
\end{itemize}

This paper derives the Ehrenfest dynamics (\ref{psi_X_eq}) and the Born-Oppenheimer approximation
from the time-independent Schr\"odinger equation (\ref{schrodinger_stat})
and %the main point here is to 
establishes
convergence rates for molecular dynamics approximations
to  time-independent Schr\"odinger observables under simple assumptions excluding so called {\it caustic} points, where
the Jacobian determinant $\det \partial X_t/\partial X_0$ of the Eulerian-Lagrangian transformation of $X-$paths vanish.
As mentioned, the main new analysis idea
%inspired by \cite{Mott,briggs,briggs2} and the semi-classical WKB analysis in \cite{maslov}, 
is  to write the time-independent Schr\"odinger equation \eqref{schrodinger_stat} as a Hamiltonian system 
and analyze the approximations  by comparing their Hamiltonians with the Schr\"odinger Hamiltonian,
using the theory of Hamilton-Jacobi partial differential equations; the problematic infinite time evolution of
perturbations in  the dynamics
is solved by viewing it as a finite time hitting  problem for the Hamilton-Jacobi equation.
Another difference is we analyze the transport equation as a time-dependent Schr\"odinger equation
in contrast to the traditional rigorous and formal asymptotic expansions.
%and approximation of the time-dependent Schr\"odinger equation  \cite{hagedorn, spohn}.  %robert}.  

%The main new analysis idea we use 
%to estimate the error between the observable of the zero-order Born-Oppenheimer dynamics
%and the observable of the time-independent Schr\"odinger equation 
%is a Hamiltonian dynamics formulation of the time-independent
%Schr\"odinger eigenfunction 
%and  the stability of the corresponding perturbed
%Hamilton-Jacobi equations viewed as a hitting problem.

%to introduce the time-dependence 
%from the classical characteristics in the Hamilton-Jacobi equation
%obtained by writing the time-independent eigenfunction (\ref{schrodinger_stat}) in WKB-form.

The main inspiration for this paper
is  \cite{Mott,briggs,briggs2} and the semi-classical WKB analysis in \cite{maslov}:
the work \cite{Mott, briggs,briggs2} 
derives the time-dependent Schr\"odinger dynamics of an $x-$system, 
$
i\dot \Psi=H_1\Psi,
$
from the time-independent Schr\"odinger equation (with the Hamiltonian $H_1(x)+ \delta H(x,X)$)
by a classical limit for the environment variable $X$, as the coupling parameter $\delta$ vanishes
and the mass $M$ tends to infinity;  in particular \cite{Mott, briggs,briggs2} show that the time derivative enters through the
coupling of $\Psi$ with the classical velocity. 
Here we refine
the use of characteristics to study  classical  {\it ab initio} molecular dynamics, where the coupling
does not vanish, and we establish error estimates for Born-Oppenheimer,
Ehrenfest, surface-hopping and stochastic approximations of the Schr\"odinger observables in the
case of no caustics present.
The small scale, introduced by the perturbation  \[
-(2M)^{-1}\sum_n\Delta_{X_n}\] of the potential $V$,
is identified in a modified WKB eikonal equation
and analyzed trough the corresponding transport equation as a time-dependent Schr\"odinger
equation along the eikonal characteristics.
This modified WKB formulation reduces to the standard semi-classical approximation, cf. \cite{maslov},
for the case of a potential function $V(X)\in\rset$ (depending
only on nuclei coordinates) but becomes different in the case of operator valued potentials studied here. 
%(not handling the two scales of mass in nuclei and electrons as presented here).
The global analysis of WKB functions started by Maslov in the 1960'  \cite{maslov} 
and lead to the subject Geometry of Quantization, relating global classical paths to eigenfunctions of
the Schr\"odinger equation, cf. \cite{duistermaat}. The analysis here,
based on a Hamiltonian system interpretation 
of the time-independent Schr\"odinger equation and stability of the corresponding Hamilton-Jacobi equation,
bypasses the usual separation 
of nuclei and electron wave functions in the time-dependent self consistent field  equations \cite{schutte,marx,tully}.

Theorems \ref{thm_ehrenfestapprox} and \ref{bo_thm}  
show that 
observables based on a {\it single WKB
Schr\"odinger eigenstate} are approximated
by observables from the Ehrenfest dynamics and the zero-order Born-Oppenheimer dynamics
with error  $\mathcal{O}(M^{-\alpha})$, 
using that these approximate solutions generate approximate eigenstates to the Schr\"odinger equation:
in the case of no caustics and the electron eigenvalues satisfying a spectral gap condition, there holds 
$\alpha=1-\delta$ for any $\delta>0$;
for non degenerate electron eigenvalue crossings the convergence rate is reduced to $\alpha=1/2-\delta$ for Born-Oppenheimer dynamics and $\alpha=3/4$ for Ehrenfest dynamics,
using the stationary phase method in
Section \ref{sec:bo}.
The results are based on the Hamiltonian (\ref{V-definition}) with any potential $V$ that is smooth in $X$,
e.g. a regularized version of the Coulomb potential (\ref{ham}).
The derivation does not  assume that the nuclei are supported on small domains;
in contrast,
derivations based on the time-dependent self consistent field equations 
require nuclei to be  
supported on small domains.
The reason that small support is not needed here comes from the combination
of the characteristics and sampling from an equilibrium density, 
that is,  the nuclei paths behave classically although they may not be supported on small domains,
in the case of no caustics. Section \ref{sec:observables} shows that caustics couple the WKB modes,
as is well known from geometric optics, see \cite{keller,maslov}, and generate non orthogonal
WKB modes that are coupled in the Schr\"odinger density;  otherwise, without caustics the Schr\"odinger
density is asymptotically decoupled into a simple sum of individual WKB densities.
%
%The derived rates of approximation using observables is higher than the previous sharp norm-based rates:
% $\mathcal{O}(M^{-1/2})$, for the Ehrenfest approximation
%\cite{schutte}, and   $\mathcal{O}(M^{-1/2})$  for the zero order Born-Oppenheimer approximation 
%\cite{spohn}, up to the time of eigenvalue crossings. 
%
Remark \ref{md_sim} relates the approximation results to 
the accuracy of symplectic numerical methods for molecular dynamics.

Section \ref{sec:time-depend} shows that the Ehrenfest 
dynamics is formally the same when derived from the time-independent
and time-dependent Schr\"odinger equations; 
a unique property of the time-independent Schr\"odinger equation we use
is the interpretation of its dynamics $X_t\in \rset^{3N}$ returning to a co-dimension one 
surface $I$ and thereby reducing the dynamics to a hitting time problem
with finite time excursions from  $I$.
Another advantage with molecular dynamics approximating  an
eigenvalue, is that stochastic perturbations of the electron ground state
can be interpreted as a Gibbs distribution of degenerate nuclei-electron eigenstates of the
Schr\"odinger eigenvalue problem (\ref{schrodinger_stat}).
The time-independent eigenvalue setting 
also avoids  the issue on "wave function collapse" to an eigenstate, present in the time-dependent Schr\"odinger equation.

The model (\ref{md_eq}) simulates dynamics
at constant energy $M|\dot X|^2/2+ V_p(X)$, constant number of particles $N$ and constant volume, i.e.
the microcanonical ensemble. The alternative to simulate with constant number of particles, constant volume
and constant {\it temperature $T$}, i.e. the canonical ensemble, is possible for instance with the stochastic
{\it Langevin dynamics}
\begin{equation}\label{lang-1}
\begin{array}{ll}
dX_\tau &= v_\tau d\tau\\
Mdv_\tau &= -\partial_X V_p(X_t)d\tau - Kv_t d\tau + (2TK)^{1/2} dW_\tau,
\end{array}
\end{equation} 
where $W_\tau$ is  the standard Brownian process (at time $\tau$)
in $\rset^{3N}$ with independent components and $K$ is a positive friction parameter. 
When the observable only depends on the nuclei positions,
i.e. not on the nuclei velocities or the correlation of positions at different times,
the {\it Smoluchowski dynamics} 
\begin{equation}\label{smoluchowski-1}
dX_\tau=- \partial_XV_p(X_\tau) + (2T)^{1/2} dW_\tau
\end{equation}
is a simplified alternative to Langevin dynamics, cf. \cite{cances}.

Section \ref{sec:s-ham} shows that the Schr\"odinger eigenvalue problem
can be written as a Hamiltonian molecular dynamics system and that
the equilibrium Gibbs distribution of the
Ehrenfest Hamiltonian system approximates the
Gibbs distribution of the Hamiltonian Schr\"odinger dynamics 
with accuracy $\mathcal O(M^{-1})$ for a spectral gap case.
Theorem \ref{thm_md_stok}
shows that Langevin and Smoluchowski dynamics accurately approximate
observables  of Schr\"odinger and Ehrenfest Gibbs equilibrium dynamics,
when the observable depends only on the nuclei positions but not their correlation at different time.
The derivation  uses a classical
equilibrium Gibbs distribution, based on a probability distribution to be in
any WKB eigenstate (corresponding to a degenerate eigenvalue of the Schr\"odinger equation (\ref{schrodinger_stat})),
randomly  perturbed from the electron ground state;
the derivation uses an other assumption of a
spectral gap and no caustics.
The main idea in the theorem is 
%the non-standard view 
the formulation of a {\it classical} Gibbs equilibrium distribution of eigenstates, motivated by
nuclei acting as heat bath for the electrons in the {\it quantum} Ehrenfest and Schr\"odinger Hamiltonian systems. 
Theorem \ref{thm3} shows that stochastic perturbations of the ground state
can also generate a temperature
dependent contribution to the  drift, depending on a spectral gap of the electron eigenvalues.

I hope that these ideas can be further developed to better understand molecular dynamics
simulations, for instance
\begin{itemize}
\item[--] I assume that the electron eigenfunction $\psi_0$ is smooth enough as a function of $X$, which holds 
when the potential $V$ is smooth in $X$ -- to find a possibly reduced
convergence rate in the  Coulomb case (\ref{ham})  will require a more careful study,
\item[--] it would be desirable to understand
the effect of caustics, as it is for systems without electrons (or with electrons as heavy as the nuclei)
in the so called semi-classical limit, cf. \cite{maslov}.
\end{itemize}

We use the notation $\psi(x,X)=\mathcal O(M^{-\alpha})$
also for complex valued functions, meaning that
$|\psi(x,X)|=\mathcal O(M^{-\alpha})$ holds uniformly in $x$ and $X$.

\section{Ehrenfest dynamics derived from the time-independent Schr\"odinger equation}
\subsection{Exact Schr\"odinger dynamics}
Assume for simplicity that all nuclei have the same mass.\footnote{If this is not the case, change to new coordinates
$M_1^{1/2}\tilde X^k=M^{1/2}_kX^k$, which transform the Hamiltonian to the form we want
${V(x,M_1^{1/2}M^{-1/2}\tilde X)}- (2M_1)^{-1}\sum_{n=1}^N\Delta_{\tilde X^n}$.}
The singular perturbation $-(2M)^{-1}\sum_n\Delta_{X_n}$ of the potential $V$ introduces
an additional  small scale $M^{-1/2}$ of high frequency oscillations,
as shown by a WKB-expansion, see   \cite{Rayleigh,jeffreys,helfer,sjostrand}.
We will construct  solutions to (\ref{schrodinger_stat}) in  such WKB-form 
\begin{equation}\label{wkb_form}
\Phi(x,X)=\psi(x,X)e^{iM^{1/2}\theta(X)},
\end{equation}
where the wave function $\psi$ is complex valued, the phase $\theta$ is real valued and the factor $M^{1/2}$
is introduced to have  well defined limits of $\psi$ and $\theta$  as  $M\rightarrow\infty$.
The standard WKB-construction  \cite{maslov,helfer} is
based on a series expansion in powers of $M^{1/2}$ which solves
the Schr\"odinger equation with arbitrary high accuracy. We introduce  instead of an asymptotic solution
an actual solution based on
a time-dependent Schr\"odinger transport equation. This transport equation
reduces to the formulation
in \cite{maslov} for the case of a potential function
$V(X)\in \rset$ (depending only on nuclei coordinates $X\in\rset^{3N}$)
and modifies it for the case of a self-adjoint potential operator $V(\cdot,X)$ on the electron space $L^2(\rset^{3J})$ focused on here; the second difference is that we analyze the transport equation
as a time-dependent Schr\"odinger equation instead of as an asymptotic expansion.
%the two scales of the electron dynamics and the nuclei dynamics. 
In the next section we  use a linear combination of such eigensolutions. 
The WKB-solution satisfies the Schr\"odinger equation (\ref{schrodinger_stat}) provided that
\begin{equation}\label{wkb_eq}
\begin{array}{ll}
0&=(H-E)\psi e^{iM^{1/2}\theta(X)}\\
&=\Big( ( \frac{|\partial_X\theta|^2}{2} + V -E) \psi\\
&\quad - \frac{1}{2M}\sum_j \Delta_{X^j}\psi
- \frac{i}{M^{1/2}}  \sum_j (\partial_{X^j} \psi\partial_{X^j}\theta
+\frac{1}{2}\psi\partial_{X^jX^j}\theta)\Big)e^{iM^{1/2}\theta(X)}\, .
\end{array}
\end{equation}
We will see that only eigensolutions $\Phi$ that correspond to
dynamics without caustics correspond to such a single WKB-mode,
as for instance when the eigenvalue $E$  is inside an electron eigenvalue gap. 
Introduce the complex-valued scalar product 
\[
v\cdot w:=\int_{\rset^{3J}} v(x,\cdot )^*w(x,\cdot) dx\equiv \langle v |w\rangle\]
on $L^2(\rset^{3J})$ and the notation $X\bullet Y$ for the standard
scalar product on $R^{3N}$.
To find an equation for $\theta$, multiply (\ref{wkb_eq}) by $\psi^* e^{-iM^{1/2}\theta(X)}$ and integrate over
$\rset^{3J}$; similarly take the complex conjugate of (\ref{wkb_eq}), multiply by $\psi e^{iM^{1/2}\theta(X)}$
and integrate over $\rset^{3J}$; and finally add the two expressions to obtain
\begin{equation}\label{theta_ekvation}
\begin{array}{ll}
0&=2\big(  \frac{|\partial_X\theta|^2}{2}  -E\big)\ \psi\cdot\psi
+\underbrace{\psi\cdot V\psi + V\psi\cdot \psi}_{=2\psi\cdot V\psi}  \\
&\quad -\frac{1}{2M}\Big( \psi\cdot (\sum_j \Delta_{X^j}\psi)
+ (\sum_j \Delta_{X^j}\psi)\cdot\psi\Big)\\
&\quad -\frac{i}{M^{1/2}}\Big(
\underbrace{\psi\cdot (\partial_{X} \psi\bullet \partial_{X}\theta)
-(\partial_{X} \psi\bullet \partial_{X}\theta)\cdot\psi}_{
=2i\Im\big(\psi \cdot(\partial_X\psi\bullet\partial_X\theta)\big)}\Big)\\
&\quad +\frac{i}{2M^{1/2}}\Big( \underbrace{(\psi\cdot\psi- \psi\cdot\psi)}_{=0}
\sum_j \partial_{X^jX^j}\theta\Big).
\end{array}
\end{equation}
The purpose of the phase function $\theta$ is to generate an accurate
approximation in the limit as $M\rightarrow\infty$: therefore we define $\theta$ by the formal limit of 
(\ref{theta_ekvation}) as $M\rightarrow\infty$, which is
the {\it Hamilton-Jacobi equation} (also called the {\it eikonal equation})
\begin{equation}\label{theta_eq}
\begin{array}{ll}
\frac{|\partial_X\theta|^2}{2} 
&=E-V_0,\\
\end{array}
\end{equation}
where the function $V_0:\rset^{3N}\rightarrow \rset$ is  
\[
V_0:=\frac{\psi\cdot V\psi}{\psi\cdot \psi}.
\]
This Eikonal equation \eqref{theta_eq} reduces to the standard formulation
\begin{equation}\label{maslov_eikon}
(\frac{|\partial_X\theta|^2}{2}   -E+V)\psi=0,
\end{equation}
cf.   \cite{maslov}, 
when the potential function $V(X)\in \rset$ depends only on nuclei 
coordinates $X\in\rset^{3N}$
but \eqref{theta_eq} is different  for the case with  an operator valued 
self-adjoint potential $V(\cdot,X)$ on the electron space $L^2(\rset^{3J})$;
in particular the standard condition \eqref{maslov_eikon} requires  $\psi$ to be an eigenfunction of $V$.
The reason we use the Eikonal equation \eqref{theta_eq} instead of \eqref{maslov_eikon}
is that the corresponding transport equation to leading order becomes
a variant of the Ehrenfest Hamiltonian dynamics \eqref{psi_X_eq}, as shown in the remainder of this section.
This section is also a step towards the more basic construction
in Section \ref{sec:s-ham}, were we modify $V_0$  in \eqref{v_0-def} with an asymptotically negligible  term,
to write the time-independent Sch\"odinger equation as
a Hamiltonian system, so that the combination of 
the Eikonal equation and the corresponding transport equation
forms a  Hamiltonian system.

%Define also the {\it density} function $\rho:=\psi\cdot\psi$.
For the energy $E$ chosen larger than the potential energy, that is such that $E\ge V_0$,  
the method of {\it characteristics}, cf. \cite{evans},
\begin{equation}\label{char_eq}
\begin{array}{ll}
\frac{d X_t}{dt} &= \partial_X\theta(X_t)=: p_t,\\
\frac{d p_t}{dt} &= -\partial_XV_0( X_t),\\
\frac{d z_t}{dt} &= | p_t|^2=2\big(E-V_0(X_t)\big),\\
\end{array}
\end{equation}
yields a solution $(X,p,z): [0,T]\rightarrow U\times\rset^{3N}\times\rset$ to the eikonal equation 
(\ref{theta_eq}) locally  in a neighborhood  $U\subseteq \rset^{3N}$,  for 
regular compatible data $( X_0,p_0, z_0)$ given on a $3N-1$ dimensional "inflow"-domain 
$ I\subset \overline{U}$;
here $ z_t:=\theta(X_t)$.
Typically the domain $I$ and the data $\theta|_I$ are not given, unless it is really an
inflow domain and characteristic paths do not return to $I$ as in a scattering problem.
If paths leaving from $I$ return to $I$, there is an additional compatibility
of data on $I$: we have $z_0= - \int_0^t |p^s|^2ds+ z_t$, where $X_0\in I$ and $X_t\in I$,
so that $z_t=\theta(X_t)$  and $p_t=\partial_X\theta(X_t)$ are determined from $z_0=\theta(X_0)$
and $p_0=\partial_X\theta(X_0)$,
and continuing the path to subsequent hitting points $X_{t_j}\in I, \ j=1,2,\ldots$ determines 
$\big(\theta(X_{t_j}),\partial_X\theta(X_{t_j})\big)$
from $\big(\theta(X_0),\partial_X\theta(X_{0})\big)$. 
Our derivation of approximation error will use such a WKB Ansatz for the approximate Ehrenfest solution, the Born-Oppenheimer solution and for the exact Schr\"odinger solution $\Phi$. 
As we shall see, the Ehrenfest and Born-Oppenheimer approximations have related (simpler) equations for
its characteristics.

%The work \cite{helfer} proves the existence of a $\mathcal{C}^\infty$ solution $\theta$
%to the eikonal equation (\ref{theta_eq})
%in a neighborhood of a global maximum point of a  given $\mathcal{C}^\infty$ non negative potential $E-V_0:\rset^{3N}\rightarrow [0,\infty)$; this handles one type of caustic (i.e. colliding characteristics) where $\partial_X\theta$ vanishes.
The phase function $\theta:U\rightarrow \rset$ becomes  globally defined  in $U\subset \rset^{3N}$ 
when there is a unique characteristic path $ X_t$ going through each point in $U$. 
A globally defined wave function $\Phi$ can be constructed from a linear combination of WKB functions 
also  when caustics are present, using  the manifold of phase-space solutions $(X,p)$ and Fourier integral operators
to relate  $X$ and $p$ dependence, 
see \cite{maslov} and \cite{duistermaat}.
We assume  in this work that characteristics do not collide
to avoid multi valued solutions. 
A simple case without caustics is when the potential is such that $\min_{X\in\rset}(E-V_0(X))>0$, in one dimension.
Section \ref{supos} presents approximations with a linear combination of WKB functions related to so called
surface-hopping.

Definition (\ref{theta_eq}) and equation \eqref{wkb_eq}
 imply that $\psi$ solves  the so called {\it transport equation}
\begin{equation}\label{psi_eq}
-\frac{1}{2M}\sum_j \Delta_{X^j}\psi
+(V-V_0)\psi= \frac{i}{M^{1/2}} \sum_j \big(\partial_{X^j}\psi\partial_{X^j}\theta 
+ \frac{1}{2} \partial_{X^jX^j} \theta\psi\big).
\end{equation}
Time enters into the Schr\"odinger equation through the characteristics and the chain rule
\[
\partial_X\psi \bullet \partial_X \theta= \partial_X\psi \bullet  \frac{d X_t}{dt}=\frac{d\psi(x,X_t)}{dt},
\]
cf. \cite{Mott,briggs,briggs2}.
The second term in the right hand side of (\ref{psi_eq})
can be simplified by the scalar integrating factor 
\begin{equation}\label{G_definition}
 G_t:=e^{\tilde G_t},
 \end{equation}
defined along the characteristics from
 \begin{equation}\label{s_int}
 \frac{d \tilde G_t}{dt}:=\frac{1}{2}\sum_j\partial_{X^jX^j}\theta(X_t).
 \end{equation}
The integrating factor $G_t$, defined by \eqref{G_definition} and \eqref{s_int},
gives 
 \[
 \begin{array}{ll}
 \sum_j \big(\partial_{X^j}\psi\partial_{X^j}\theta 
+ \frac{1}{2} \partial_{X^jX^j} \theta\psi\big)
&=\frac{d\psi}{dt} +\frac{\psi}{G_t}\frac{dG_t}{dt}\\
&=\frac{1}{G_t}\frac{d(G_t\psi)}{dt}.\\
\end{array}
\]
We can also define  a function $G:\rset^{3N}\rightarrow \rset$ by identifying $G(X_t):=G_t$.
This step, with the integrating factor $G$,
 differs from  \cite{Mott,briggs,briggs2}, which approximates the last term in (\ref{psi_eq}),
\[
\sum_j\partial_{X^jX^j}\theta\psi,
\]
 by zero in their case of vanishing coupling between the quantum system and the environment; here the
 coupling between the nuclei and electrons does not vanish.
The right hand side in \eqref{psi_eq} becomes the time derivative
$iM^{-1/2}G^{-1}\ d(G\psi)/dt$ and we have derived the {\it time-dependent  Schr\"odinger} equation, for the
variable $\tilde\psi:=G\psi$,
\begin{equation}\label{psi_first_eq}
\begin{array}{ll}
0&=(H-E)\Phi\\
&=\Big(\big(-\frac{i}{M^{1/2}}\dot{ \tilde\psi}+ (V-V_0)\tilde\psi 
-\frac{1}{2M}\sum_j G\Delta_{X^j}(\tilde\psi G^{-1})\big) G^{-1}\\
&\quad
+\underbrace{\big( \frac{|\partial_X\theta|^2}{2} + V_0 -E\big)}_{=0} \psi\Big)e^{iM^{1/2}\theta(X)}\, .
\end{array}
\end{equation}
The density can be written
\begin{equation}\label{rho_definition}
\rho:= \int_{\rset^{3J}} |\psi|^2 dx=\int_{\rset^{3J}} |\tilde\psi|^2 dx/ G^2
\end{equation}
and therefore the second equation in \eqref{char_eq} yields the nuclei dynamics 
\[
\ddot{ X}= - \partial_X\frac{\tilde \psi\cdot V \tilde \psi}{ {\tilde\psi\cdot \tilde\psi}}.
\]
The weight function $G^2_t$  equals the determinant of the first variation $\partial X_t/\partial X_0$ modulo a
constant
\begin{equation}\label{euler_lagrange-det}
G^2_t/G^2_0=\mbox{det} (\partial X_t/\partial X_0),
\end{equation}
which follows from Liouville's  formula,
see
\cite{maslov},
and in one dimension $G^2_t/G_0^2=p_t/p_0$.

In conclusion, we have rewritten a WKB solution of the time-independent Schr\"odinger
equation in the form of  the exact {\it Schr\"odinger dynamics} along the characteristics
of the Eikonal equation and the transport equation
\begin{equation}\label{tilde_psi_eq}
\begin{array}{ll}
&\frac{i}{M^{1/2}}\dot{\tilde{\psi}} =(V-V_0)\tilde\psi-\frac{1}{2M}\sum_j G\Delta_{X^j}(\tilde\psi/G),\\
&\ddot{X} = -  \partial_X\frac{\tilde \psi\cdot V\tilde \psi}{ {\tilde\psi\cdot \tilde\psi}}\, .
\end{array}
\end{equation}
The eigenvalue $E$ is a parameter in the Hamiltonian for the characteristics  
of the Hamilton-Jacobi equation \eqref{theta_eq}.
 In the case when no electrons are present, we have $V=V_0$, and equation \eqref{tilde_psi_eq} 
 was derived from (\ref{schrodinger_stat}-\ref{V-definition}) and \eqref{wkb_form} in \cite{maslov}.
The integrating factor $G$ and its derivative $\partial_XG$ can be determined from $(p,\partial_Xp,\partial_{XX}p)$
along the characteristics by
the following characteristic equations obtained from $X$-differentiation of \eqref{theta_eq}
\begin{equation}\label{p_xx_eq}
\begin{array}{ll}
\frac{d}{dt}{\partial_{X^r} p^k}& =\Big[ \sum_j p^j\partial_{X^jX^r}p^k=\sum_jp^j\partial_{X^rX^k}p^j\Big]\\
&=  -\sum_j\partial_{X^r} p^j\partial_{X^k} p^j -\partial_{X^rX^k}V_0,\\
\frac{d}{dt}{\partial_{X^rX^q} p^k}&= \Big[\sum_j p^j\partial_{X^jX^rX^q}p^k=\sum_jp^j\partial_{X^rX^kX^q}p^j\Big]\\
&=  -\sum_j\partial_{X^r} p^j\partial_{X^kX^q} p^j 
-\sum_j\partial_{X^rX^q} p^j\partial_{X^k} p^j -\partial_{X^rX^kX^q}V_0,\\
\end{array}
\end{equation}
and similarly $\partial_{XX}G$ can be determined from  $(p,\partial_Xp,\partial_{XX}p,\partial_{XXX}p)$.

\subsection{Approximate Ehrenfest dynamics and densities}
We define the approximating {\it Ehrenfest dynamics} by  in \eqref{wkb_eq}
neglecting  the term $(2M)^{-1} \sum_j\Delta_{X^j}\psi$:
\begin{equation}\label{hat_theta_eq}
\begin{array}{ll}
0=\Big( \frac{|\partial_X\hat\theta|^2}{2} + V -E\Big) \check\psi
 - \frac{i}{M^{1/2}}  \sum_j (\partial_{X^j} \check\psi\partial_{X^j}\hat\theta
+\frac{1}{2}\check\psi\partial_{X^jX^j}\hat\theta),
\end{array}
\end{equation}
and seek, as in \eqref{theta_eq}, the approximate  phase $\hat\theta$
as the solution to eikonal equation
\begin{equation}\label{hj_classic}
 \frac{|\partial_X\hat\theta|^2}{2} =E- \hat V_0,
\end{equation}
where 
\[
\hat V_0:=\frac{\check\psi\cdot V\check \psi}{\check\psi\cdot \check \psi}.
\]
Introduce its characteristics 
\[
\begin{array}{ll}
\frac{d\hat X_t}{dt} &= \partial_X\hat\theta(X_t)=:\hat p_t,\\
\frac{d\hat p_t}{dt} &= -\partial_X\hat V_0(\hat X_t),\\
\frac{d\hat z_t}{dt} &= |\hat p_t|^2=2\big(E-\hat V_0(\hat X_t)\big),\\
\end{array}
\]
to rewrite \eqref{hat_theta_eq}, as in \eqref{psi_first_eq},
\begin{equation}\label{ehrenfest_dyn}
\begin{array}{ll}
0&=-\big(\frac{i}{M^{1/2}}\dot{\hat\psi} -(V-\hat V_0)\hat\psi\big)\hat G^{-1} 
+ \underbrace{\big( \frac{|\partial_X\hat\theta|^2}{2} + \hat V_0 -E\big)}_{=0} \check\psi ,\\
\ddot{\hat X} &= -\partial_X \Big(\frac{\hat\psi \cdot V\hat\psi}{\hat\psi\cdot \hat\psi}\Big),
\end{array}
\end{equation}
for $\hat\psi:=\hat G\check\psi$  approximating $\tilde\psi$ and  
\begin{equation}\label{G_const}
\hat G_t:= Ce^{\int_0^t2^{-1}\sum_j\partial_{X^j}\hat p^j_s ds}\end{equation}
as in \eqref{s_int} (where $C$ is a positive constant for each characteristic).
Using that  $\hat\psi\cdot \hat\psi$ is conserved (i.e. time-independent)
in the Ehrenfest dynamics, we can normalize to $\hat\psi\cdot\hat\psi=1$.
 Note that in the exact dynamics, the function $\tilde\psi\cdot\tilde\psi$ is not conserved,
due to the  $L^2(\rset^{3J})$ non symmetric source term $\frac{1}{2M}\sum_j G\Delta_{X^j}(\tilde\psi/G)$ in \eqref{tilde_psi_eq}.
We have
\begin{equation}\label{v0-def}
V_0= \frac{\psi\cdot V\psi}{\psi\cdot\psi}
= \frac{\tilde\psi\cdot V\tilde\psi}{\tilde\psi\cdot\tilde\psi}
\end{equation}
and
\[
\hat V_0 =\frac{\hat\psi\cdot V\hat\psi}{\hat\psi\cdot\hat\psi}=\hat\psi\cdot V\hat\psi.
\]

\subsection{Comparison of two alternative Ehrenfest formulations}\label{ehrenfest_comparison}
The two different Ehrenfest dynamics \eqref{ehrenfest_dyn} respectively  \eqref{psi_X_eq}  
differ in: 
\begin{itemize}
\item[(1)] the different time scales  $t$ (slow) respectively  $M^{1/2}t=:\tau$ (fast); 
\item[(2)] the potentials $V-\hat V_0$ and  $V$ 
in the equations for $\hat\psi$ and $\phi$, respectively; and
\item[(3)] the forces $\partial_X(\hat\psi\cdot V\hat\psi)$ respectively $\phi\cdot \partial_XV\phi$
in the momentum equation.
\end{itemize}
If $\hat\psi$ solves \eqref{ehrenfest_dyn}, the change of variables  
\[
\hat\phi=\hat\psi e^{-iM^{1/2}\int_0^t\hat\psi\cdot V\hat\psi(X_s)ds}\]
and the property  $\hat\psi\cdot A\hat\psi =\hat\phi\cdot A\hat\phi $, which holds for observables $A$ 
not including the nuclei momentum $i\partial_X$
(in particular for $A=V$),
imply that $\hat\phi$ solves 
\begin{equation}\label{transform_phi}
\begin{array}{ll}
\frac{i}{M^{1/2}}\dot{\hat\phi}& = V\hat\phi,\\
\\
\ddot{\hat X} &= -\partial_X \Big(\frac{\hat\phi \cdot V\hat\phi}{\hat\phi\cdot \hat\phi}\Big)
=\partial_X (\hat\phi \cdot V\hat\phi).
\end{array}
\end{equation}
%so that $\phi(\tau)=\hat\phi(t)$ solves $i\dot\phi=V\phi$.

There has been a discussion in the literature \cite{schutte,tully} 
whether the forces should be computed as above in \eqref{transform_phi} or
as in \eqref{psi_X_eq} by
\begin{equation}\label{hj-system}
\begin{array}{ll}
M\ddot {X}^n_\tau&= -\int_{\rset^{3J}} 
\phi^*(\cdot, X_\tau) \,  {\partial_{X^n}V(\cdot,X_\tau)}\, \phi(\cdot, X_\tau)\ dx \\
i\dot \phi_\tau &=V(\cdot, X_\tau)\phi_\tau.
\end{array}
\end{equation}
Theorems \ref{thm_ehrenfestapprox}  and \ref{bo_thm} below  show that  both  formulations \eqref{transform_phi}  
and (\ref{hj-system}, \ref{psi_X_eq}) yield 
accurate approximations of the Schr\"odinger observables, although they are not the same and (\ref{hj-system}, \ref{psi_X_eq}) is closer to the Schr\"odinger observables:
the reason that both approximations are accurate is that 
 (\ref{hj-system}, \ref{psi_X_eq}) forms a Hamiltonian system (as explained below and in Section \ref{sec_hj}) 
 which is close the 
 Hamiltonian system for the Schr\"odinger equation, see Theorem  \ref{thm_ehrenfestapprox}, while
 the formulation \eqref{transform_phi} can be viewed as a Hamiltonian system approximating the Born-Oppenheimer dynamics \eqref{md_eq}, using that the wave function $\check\varphi$ is close to an electron eigenfunction (see Section \ref{sec:bo}),
and the Born-Oppenheimer dynamics approximates Schr\"odinger observables, see Theorem \ref{bo_thm}.

%Theorem \ref{thm_ehrenfestapprox} below  shows that  both  formulations \eqref{transform_phi}  
%and (\ref{hj-system}, \ref{psi_X_eq}) yield 
%accurate approximations of the Schr\"odinger observables, although they are not the same.
%The reason that both approximations are accurate is that the two different characteristic systems
%solve the same Hamilton-Jacobi equation, as explained below and in Section \ref{sec_hj}.

The formulation  \eqref{psi_X_eq} and \eqref{hj-system}
has the advantage
to be a closed Hamiltonian system:   the variable $(X,\varphi_r;p,\varphi_i)$, with the definition
\[
\begin{array}{ll}
\varphi&:=\varphi_r+ i\varphi_i:= 2^{1/2}M^{-1/4}\phi= 2^{1/2}M^{-1/4}(\phi_r+i\phi_i),\\
\partial_{\varphi_r}\tilde\theta&=: \varphi_i,
\end{array}
\]
and the {\it Ansatz}  $\tilde\theta=\hat\theta$ imply that the  Hamilton-Jacobi equation \eqref{hj_classic} becomes
\begin{equation}\label{E-hamiltonian}
\begin{array}{ll}
H_E&:=\frac{1}{2}\partial_X\tilde\theta\bullet\partial_X\tilde\theta 
+ \phi_r\cdot V(X)\phi_r + \phi_i\cdot V(X)\phi_i \\
&=\frac{1}{2}\partial_X\tilde\theta\bullet\partial_X\tilde\theta 
+ 2^{-1}M^{1/2} \varphi_r\cdot V(X)\varphi_r + 2^{-1}M^{1/2}\varphi_i\cdot V(X)\varphi_i\\
&=
\frac{1}{2}\partial_X\tilde\theta\bullet\partial_X\tilde\theta 
+ 2^{-1}M^{1/2} \varphi_r\cdot V(X)\varphi_r 
+ 2^{-1}M^{1/2} \partial_{\varphi_r}\tilde\theta \cdot V(X)\partial_{\varphi_r}\tilde\theta\\
&=E,
\end{array}
\end{equation}
where the derivative 
$\partial_{\varphi_r}\tilde\theta(X,\varphi_r)=\varphi_i$, 
of the functional $\tilde\theta:\rset^{3N}\times L^2(dx)\rightarrow\rset$,
is the Gateaux derivative in $L^2(dx)$.
Its characteristics form
the { Hamiltonian system}
\[
\begin{array}{ll}
\dot X_t &=p_t\\
\dot p_t &=  -\frac{M^{1/2}}{2} \varphi(t)\cdot \partial_X V(X_t)\varphi(t)\\
\dot \varphi_r(t) &=  {M^{1/2}}V(X_t)\varphi_i(t)\\
\dot \varphi_i(t) &=  -{M^{1/2}} V(X_t)\varphi_r(t),\\
\end{array}
\] which is the same as the {\it Hamiltonian system} \eqref{hj-system}
\begin{equation}\label{ehrenfest_ham}
\begin{array}{ll}
\dot{X}_t &=p_t\\
\dot{p}_t &=   -\phi_t\cdot \partial_X V(X_t)\phi_t\\
\frac{i}{M^{1/2}}\dot{\phi}_t &=   V(X_t)\phi_t.\\
\end{array}
\end{equation}
The Hamiltonian system yields the  equation for the phase
\[
\dot{\tilde \theta}= \partial_X\tilde\theta\bullet \dot X+ \partial_{\varphi_r}\tilde\theta\cdot \dot\varphi_r
= p\bullet p + 2 \phi_i\cdot V\phi_i 
%= E + (\frac{1}{2}p\bullet p  + \phi_i\cdot V\phi_i - \phi_r\cdot V\phi_r)
=2(E- \phi_r\cdot V\phi_r),
\]
since $\tilde\theta=\tilde\theta(X,\varphi_r)$ is a function of both $X$ and $\varphi_r$ in this formulation.
%and we see that $\tilde\theta$ is indeed a function of $X$, independent of $x$
%explicitly (but dependent on $\phi$ which depends on $x$ implicitly)
%so that the identification of the values
%$\hat\theta(X)=\tilde\theta(X, \varphi_r(X))$ is possible. 
The important property of this Hamiltonian dynamics is that $(X,p,\hat\psi)$, 
with \[\hat\psi=\phi e^{iM^{1/2}\int_0^t \phi\cdot V\phi(s) ds},\] solves both
the Hamilton-Jacobi equation \eqref{hj_classic} and 
\eqref{hat_theta_eq}, %written as the time-dependent Schr\"odinger equation \eqref{ehrenfest_dyn},  
which leads to an approximate solution of the Schr\"odinger eigenvalue problem in 
\eqref{ehren_dynamik}-\eqref{h-e-approx}.

The alternative \eqref{transform_phi} does not form a closed system, in the sense that
the required function $\partial_X\psi(X_t)$ is not explicitly determined along the characteristics, but
can be obtained from values of $\psi$, in a neighborhood of $X_t$, by differentiation.

\subsection{Equations for the density}
We note that
\[
\begin{array}{ll}
\psi&=\Big(\frac{\rho}{\tilde\psi\cdot\tilde\psi}\Big)^{1/2}\, \tilde\psi,\\
\check\psi&=\Big(\frac{\hat\rho}{\hat\psi\cdot\hat\psi}\Big)^{1/2}\, \hat\psi,
\end{array}
\]
shows that 
the densities $\rho=\psi\cdot\psi$ (defined in \eqref{rho_definition})
and $\hat\rho:=\check\psi\cdot\check\psi$,
in addition to $(X,p,\tilde\psi)$ and  $(\hat X,\hat p,\hat\psi)$,
are needed
to determine the wave functions $\psi$ and $\check\psi $.

 Equation \eqref{wkb_eq} 
 and the projections in \eqref{theta_ekvation} subtracted (instead of added) 
imply that the density $\rho$ satisfies
\[
\begin{array}{ll}
-M^{-1/2}\sum_j \Im(\psi\cdot \Delta_{X^j}\psi)&=\sum_j\int_{\rset^{3J}} (\partial_{X^j}\psi^*\psi 
+ \psi^*\partial_{X^j}\psi)dx\ \partial_{X^j}\theta 
+\int_{\rset^{3J}}\psi^*\psi dx\ \partial_{X^jX^j}\theta\\
&= \sum_j\partial_{X^j}(\rho\partial_{X^j}\theta),
\end{array}
\]
where $\Im w$ denotes the imaginary part of $w$,
and consequently, the density can be determined along a characteristic
using \eqref{s_int} and \eqref{p_xx_eq}
\begin{equation}\label{density_eq}
\begin{array}{ll}
\dot{\rho}( X_t)&=\sum_j\partial_{X^j}\rho(X_t)\dot{X}^j\\
&=\sum_j\partial_{ X^j}\rho( X_t)\partial_{\hat X^j}\theta\\
&= -\rho(X_t)\sum_j\partial_{ X^j X^j}\theta 
-M^{-1/2}\sum_j \Im(\psi\cdot \Delta_{X^j}\psi)\\
&= -\rho( X_t)\ \mbox{div} -M^{-1/2}\sum_j \Im(\psi\cdot \Delta_{X^j}\psi)\,  p\\
&= -\rho( X_t)\frac{d}{dt}\log  G_t^2-M^{-1/2}\sum_j \Im(\psi\cdot \Delta_{X^j}\psi).
\end{array}
\end{equation}

% Equation \eqref{hat_theta_eq} 
% and the projections in \eqref{theta_ekvation} subtracted (instead of added) 
%imply that the approximate density $\hat\rho$ satisfies
%\[
%\begin{array}{ll}
%0&=\sum_j\int_{\rset^{3J}} (\partial_{X^j}\check\psi^*\check\psi 
%+ \check\psi^*\partial_{X^j}\check\psi)dx\ \partial_{X^j}\hat\theta 
%+\int_{\rset^{3J}}\check\psi^*\check\psi dx\ \partial_{X^jX^j}\hat\theta\\
%&= \sum_j\partial_{X^j}(\hat\rho\partial_{X^j}\hat\theta)
%\end{array}
%\]
%and consequently, the density can be determined along a characteristic
%using \eqref{s_int} and \eqref{p_xx_eq}
%\begin{equation}\label{approx_dens}
%\begin{array}{ll}
%\dot{\hat\rho}(\hat X_t)&=\sum_j\partial_{\hat X^j}\hat\rho(\hat X_t)\dot{\hat X}^j\\
%&=\sum_j\partial_{\hat X^j}\hat\rho(\hat X_t)\partial_{\hat X^j}\hat\theta\\
%&= -\hat\rho(\hat X_t)\sum_j\partial_{\hat X^j\hat X^j}\hat\theta\\
%&= -\hat\rho(\hat X_t)\ \mbox{div} \, \hat p\\
%&= -\hat\rho(\hat X_t)\frac{d}{dt}\log \hat G_t^2
%\end{array}
%\end{equation}
%with the solution
%\begin{equation}\label{density_fact}
%\hat\rho(\hat X_t)=\frac{C}{\hat G_t^2},
%\end{equation}
%where $C$ is a positive  constant for each characteristic. Note that this density corresponds precisely to
%the Eulerian-Lagrangian change of coordinates $\hat G^2_t/\hat G^2_0=\mbox{det} (\partial \hat X_t/\partial \hat X_0)$
%in \eqref{euler_lagrange-det}.

Similarly, the Ehrenfest  density satisfies
the conservation of mass
\[
0 %=\sum_j\int_{\rset^{3J}} (\partial_{X^j}\check\psi^*\check\psi 
%+ \check\psi^*\partial_{X^j}\check\psi)dx\ \partial_{X^j}\hat\theta 
%+\int_{\rset^{3J}}\check\psi^*\check\psi dx\ \partial_{X^jX^j}\hat\theta\\
= \sum_j\partial_{\hat X^j}(\hat\rho\partial_{\hat X^j}\hat\theta)
%\mbox{div}(\hat\rho \partial_{\hat X}\hat\theta)=0
\]
so that
\begin{equation}\label{approx_dens}
\begin{array}{ll}
\dot{\hat\rho}( \hat X_t)&=\sum_j\partial_{\hat X^j}\hat\rho(\hat X_t)\dot{\hat X}^j\\
&= -\hat\rho( \hat X_t) \frac{d}{dt}\log \hat G_t^2 
%+ M^{-1/2}\sum_j \Im(\psi\cdot \Delta_{X^j}\psi),
\end{array}
\end{equation}
with the solution
\begin{equation}\label{density_fact}
\hat\rho(\hat X_t)=\frac{C}{\hat G_t^2},
\end{equation}
where $C$ is a positive  constant for each characteristic. Note that the derivation of this classical density 
does not need a corresponding WKB equation but
uses only the conservation of mass 
that holds for classical paths.
The classical density corresponds precisely to
the Eulerian-Lagrangian change of coordinates $\hat G^2_t/\hat G^2_0=\mbox{det} (\partial \hat X_t/\partial \hat X_0)$
in \eqref{euler_lagrange-det}.
Different characteristic paths 
$\hat X$
may have different densities when a path does not visit the whole configuration space $\rset^{3N}$.
The density from the Schr\"odinger equation is therefore important to weight  different paths.

\subsubsection{ The $\rho-\hat \rho$ error without electrons.} \label{wkb_utan_elektroner}
Comparing \eqref{density_eq} and \eqref{approx_dens}, 
the difference $\rho-\hat\rho$ has contributions both from
$G-\hat G$ and from the 
error term $M^{-1/2}\sum_j \Im(\psi\cdot \Delta_{X^j}\psi)$.
In this section we show heuristically how the characteristics can be used 
to estimate the difference $\rho-\hat\rho$, leading to
 $\mathcal{O}(M^{-1})$ accurate Ehrenfest
approximations of Schr\"odinger observables
\[
\int g(X)\underbrace{\rho(X)}_{\Phi\cdot\Phi}dX=\int g(X)\hat\rho(X)dX +\mathcal{O}(M^{-1}),
\]
in the well known case of no electrons present; Section \ref{sec:bo} extends this derivation to the case including electrons.

In the special case of no electrons,  the $X$ dynamics does not depend on $\tilde\psi$
and therefore $\hat X=X$ and consequently $\hat G=G$.
The difference $\tilde\psi-\hat\psi $ can be
understood from iterative approximations of \eqref{tilde_psi_eq}
\begin{equation}\label{hat_noel}
\frac{i}{M^{1/2}} \frac{d}{dt}{\hat\psi}_{n+1} -(V-\hat V_0)\hat\psi_{n+1}=
\frac{1}{2M} G\sum_j\Delta_{X^j}(G^{-1}\hat\psi_n)
\end{equation}
with $\hat\psi_0=0$.
Then $\hat\psi_1=\hat\psi$ is the Ehrenfest approximation
and formally we have the iterations approaching the
full Schr\"odinger solution $\hat\psi_n\rightarrow \tilde\psi$ as $n\rightarrow \infty$.

In the special case of no electrons, there holds $V=\hat V_0$, so the transport equation
$i\dot{\hat{\psi}}_1=0$ has constant solutions; let $\hat\psi_1=1$.
Then $\hat\psi_2-\hat\psi_1$ is imaginary with its absolute value bounded by $\mathcal{O}(M^{-1/2})$;
write the iterations of $\hat \psi_n$ by integrating \eqref{hat_noel} as the linear mapping
\[
\hat\psi_{n+1}=1 +iM^{-1/2}S(\hat\psi_n)
=\sum_{k=0}^{n} i^kM^{-k/2}S^k(\hat\psi_1)
\]
which formally shows that 
\[
|\tilde\psi |^2=|\hat\psi_1|^2 + 2\Re\big((\tilde\psi-\hat\psi_1)\cdot \hat\psi_1) + |\tilde\psi-\hat\psi_1|^2
=1+\mathcal{O}(M^{-1}).
\]
Consequently this special Ehrenfest density satisfies 
\begin{equation}\label{g_dens}
\hat\rho=\hat G^{-2}=\underbrace{G^{-2} \tilde\psi\cdot\tilde\psi}_{=\rho}+\mathcal{O}(M^{-1}),
\end{equation}
since $\hat G=G$ and  $X$ do not depend on $\tilde\psi$. 
 In the general case with electrons,  we show in Section \ref{sec:bo} 
 that there is a solution $\tilde\psi$ which is $\mathcal{O}(M^{-1/2})$
close in $L^2(dx)$ to an electron eigenfunction $\psi_0$, satisfying 
\[ %\begin{equation}\label{eigen_def}
V(X,\cdot)\psi_0(X,\cdot)=\lambda_0(X)\psi_0(X,\cdot)
\] %\end{equation}
for an eigenvalue $\lambda_0(X)\in\rset$ and (fixed) nuclei position $X$. 
The state $\hat\psi_1$ equal to a constant, in the case of no electrons,
corresponds to the electron eigenfunction $\psi_0$ in the case with electrons present.
In the general case with electrons, Section \ref{sec:bo}
shows that  still
\begin{equation}\label{alfa_ekv}
\tilde\psi-\hat \psi= i\alpha +\mathcal O(M^{-1})
\end{equation}
where $\alpha=\mathcal O(M^{-1/2})$ is real and parallel to $\psi_0$, which implies that the Hamiltonians $H_S$ and $H_E$ are
%$X$ and $\hat X$ 
are $\mathcal O(M^{-1})$ close so also $G-\hat G=\mathcal O(M^{-1})$
and consequently the density bound \eqref{g_dens} holds.
To obtain the estimate \eqref{alfa_ekv} the important new property is that
oscillatory cancellation is used  in directions orthogonal to $\psi_0$,
reducing the error terms from $\mathcal O(M^{-1/2})$ to $\mathcal O(M^{-1})$,
in the case when a spectral gap condition holds.

\subsection{Construction of the solution operator}\label{start_sec}
The WKB-forms  \eqref{wkb_form}  and \eqref{hat_theta_eq} are meaningful when $\psi$ and $\check\psi$
do not include the full small scale
%in our theorems we will in fact only need this property 
%for the  Ehrenfest solution $\check\psi$
%but it is helpful for the understanding to also consider the function $\psi$
and we verify  in Section \ref{born_oppen} 
that both $\partial_X \psi$ and $\partial_X\check\psi$ are bounded independent of $M$, in the case of a spectral gap.
%using  a spectral representation. In this section we present the set-up.
 Section \ref{born_oppen} also presents  conditions so
 that $\tilde\psi$ is $\mathcal{O}(M^{-1/2})$ close to an eigenvector  of $V$ in $L^2(dx)$.
To replace $\tilde\psi$
by such an electron eigenstate 
is called the {\it Born-Oppenheimer approximation}, which has been
studied  for the time-independent \cite{tanner,BO} and the time-dependent \cite{hagedorn,spohn} Schr\"odinger 
equations  by different methods.

To construct the solution operator it is convenient  to include a non interacting particle, i.e. a particle without charge,
in the system and assume that this particle moves with constant  high speed $dX_1^1/dt=p_1^1\gg 1$ (or equivalently
with speed one and larger mass);
such a non interacting particle does clearly not effect the other particles. The additional new coordinate $X^1_1$
is helpful in order to simply relate the  time-coordinate $t$ and $X^1_1$. To not change the original problem
\eqref{schrodinger_stat}, add the corresponding kinetic energy $(p^1_1)^2/2$ to $E$ and
write equation \eqref{psi_first_eq} in the fast time scale $\tau=M^{1/2}t$
\[
i \frac{d}{d\tau}\tilde\psi=(V-V_0)\tilde\psi - \frac{1}{2M} G\sum_j\Delta_{X^j}(G^{-1}\tilde\psi)
\]
and change to the coordinates 
\[
\mbox{$(\tau,X_0):=(\tau,X^1_2,X^1_3,X^2,\ldots,X^N)\in [0,\infty)\times I$ 
instead of $(X^1,X^2,\ldots, X^N)\in\rset^{3N},$}
\]
where $X^j=(X^j_1,X^j_2,X^j_3)\in \rset^3$,  to obtain
\begin{equation}\label{new_psi_eq}
\begin{array}{ll}
i \dot{\tilde\psi} +\frac{1}{2(p_1^1)^2}\ddot{\tilde\psi}
&=(V-V_0)\tilde\psi - \frac{1}{2M} G\sum_j\Delta_{X_0^j}(G^{-1}\tilde\psi)\\
&=:\tilde V\tilde\psi,
\end{array}
\end{equation}
using the notation $\dot w= dw/d\tau$ in this section; note also that $G$ is independent of $X_1^1$.
We see that the operator
\[
\bar V:=G^{-1} \tilde V G=\underbrace{G^{-1} (V-V_0)G}_{=V-V_0} -\frac{1}{2M}\sum_{j}\Delta_{X^j_0}
\]
is symmetric on $L^2(\rset^{3J+3N-1})$. The solution of the eikonal equation \eqref{theta_eq},
by the characteristics \eqref{char_eq}, becomes well defined
in a domain $U=[0,M^{1/2} t]\times \rset^{3N-1}$, in the new coordinates. 
Assume now the data $(X_0,p_0,z_0)$
for $X_0\in \rset^{3N-1}$ is $(L\mathbb Z)^{3N-1}$-periodic, then also $(X_\tau,p_\tau,z_\tau)$
is $(L\mathbb Z)^{3N-1}$-periodic.
To simplify the notation for such periodic functions, define the periodic circle \[\tset:=\rset/(L\mathbb Z).\]
We seek a solution $\Phi$ of \eqref{schrodinger_stat}
which is $(L\mathbb Z)^{3(J+N)-1}$-periodic in the $(x,X_0)$-variable.
The  Schr\"odinger operator $\bar V_\tau$ 
has, for each $\tau$, real eigenvalues $\{\lambda_m(\tau)\}$
with a complete set of eigenvectors $\{p_m(x,X_0)_\tau\}$ orthogonal in the space
of $x$-anti-symmetric functions in $L^2(\mathbb T^{3J+3N-1})$, see \cite{berezin};
its proof uses
that the operator $\bar V_\tau+\gamma I$ generates a compact solution operator
in the Hilbert space
of $x$-anti-symmetric functions in $L^2(\mathbb T^{3J+3N-1})$, for the constant $\gamma\in (0,\infty)$
chosen sufficiently large. The discrete spectrum and the compactness comes from 
Fredholm theory for compact operators and 
that the bilinear form $\int_{\tset^{3(J+N)-1}} v \bar V_\tau w +\gamma vw \ dxdX_0$
is continuous and coercive on $H^1(\tset^{3(J+N)-1})$, see \cite{evans} and Section \ref{fredholm}.
We see that  $\tilde V$ has the same eigenvalues $\{\lambda_m(\tau)\}$ and the eigenvectors $\{G_\tau p_m(\tau)\}$, 
orthogonal in the weighted $L^2$-scalar product 
\[
%v\bullet w:=
\int_{\mathbb T^{3N-1}}v \cdot w \ G^{-2} dX_0.\]
The construction and analysis of the solution operator continues in Sections \ref{born_oppen} and \ref{fredholm}
based on the spectrum.

\begin{remark}[{Boundary conditions}]
The eigenvalue problem \eqref{schrodinger_stat} makes sense not only
in the periodic setting but also with alternative {boundary conditions}
from interaction with an external environment, e.g. for scattering problems.
The inflow, with data given from
the time-independent Schr\"odinger problem, and the outflow of
characteristics gives a different perspective
on molecular dynamics simulations 
and the possible initial data for the time-dependent Schr\"odinger equation.
\end{remark}

\section{The time-dependent Schr\"odinger equation}\label{sec:time-depend}
The corresponding time-dependent {\it Ansatz} 
\[\psi(x,X,t)e^{iM^{1/2}\big(\theta(X,t)+Et\big)}\]
in the time-dependent Schr\"odinger equation  \cite{schrodinger}
\begin{equation}\label{time_dep_qm}
\frac{i}{M^{1/2}}\dot \Phi=H\Phi
\end{equation}
leads analogously to the equations 
\[
\begin{array}{ll}
\partial_t \theta + \frac{|\partial_X\theta|^2}{2}  &= E-V_0,\\
\partial_t \rho + \sum_j\partial_{X^j}(\rho\partial_{X^j}\theta) 
&=  M^{-1/2}{\sum_j \Im( \psi\cdot \Delta_{X^j}\psi }),\\
\end{array}
\]
coupled to the Schr\"odinger equation along its  characteristics $X_t$
\begin{equation}\label{time_char}
\begin{array}{ll}
&\frac{i}{M^{1/2}}\frac{d}{dt}\tilde{\psi}(x,X_t,t) =(V-V_0)\tilde\psi
-  \frac{G}{2M}{ \sum_j \Delta_{X^j}(\tilde\psi G^{-1})},\\
&\ddot{X} = - \frac{\partial_X (\tilde \psi\cdot V\tilde  \psi) }{\tilde\psi\cdot\tilde\psi},
\end{array}
\end{equation}
with the same equation for $(X,p,\tilde\psi,\rho,z,\partial_Xp,\partial_{XX}p)$
as for the characteristics  \eqref{theta_eq} and \eqref{tilde_psi_eq} in the time-independent formulation.
The Ehrenfest dynamics is therefore the same
when derived from the time-dependent and time-independent Schr\"odinger equations
and the additional coordinate, introduced by a non interacting particle in the construction of 
the time-independent solution in Section \ref{start_sec}, can be interpreted as time.
A difference is that the time variable is given
from the time-dependent  Schr\"odinger equation and implies classical velocity,
instead of the other way around for the time-independent formulation.
 
 \section{Surface-hopping and multiple states}\label{sec:surf}
In general the eigenvalue $E$ is degenerate with high multiplicity related to
that several combinations of kinetic nuclei energy and potential energy sum to $E$ 
\[
-\int_{\tset^{3(J+N)}}\Phi^*
\frac{1}{2M}\sum_j\Delta_{X^j}\Phi dxdX 
+\int_{\tset^{3(J+N)}} \Phi^*V\Phi dx dX =E\int_{\tset^{3(J+N)}} \Phi^*\Phi dx dX,
\]
with different excitations of kinetic nuclei energy and Born-Oppenheimer electronic eigenstates. 
When several such states are excited, it is
useful to consider a linear combination  of eigenfunctions $\sum_{n=1}^{\bar n}
\psi_n e^{iM^{1/2}\theta_n}=:\bar\Phi$, 
where the individual terms solve \eqref{schrodinger_stat}  for the same energy $E$.
We have
\begin{equation}\label{phi_ekvation}
H\bar\Phi=E\bar\Phi
\end{equation}
and the normalization $\sum_{n=1}^{\bar n}\|\psi_n\|_{L^2}^2=1$   implies $\|\bar\Phi\|_{L^2}=1$.
Such a solution $\bar\Phi$ can be interpreted 
as an  exact {\it surface-hopping} model.
The usual surface-hopping models make a somewhat different {\it Ansatz} with
the $x$-dependence of $\psi_n$  prescribed from a given  orthonormal basis in $L(\tset^{3J})$ of wave functions
of different energy
and with explicit time-dependence, see \cite{tully,tanner}.
This section extends the Ehrenfest dynamics to multiple states.

\subsection{Surface-hopping and Ehrenfest dynamics for multiple states}\label{supos}
The characteristic $(X_n,p_n)$, the wave function $\tilde\psi_n$, the density 
$\rho_n$ and the phase $z_n(t)=\theta_n\big(X_n(t)\big)$ 
determine the time-independent wave function $\psi_n$
and the corresponding Ehrenfest approximation
\begin{equation}\label{surf_hopp}
\begin{array}{ll}
\sum_{n=1}^{\bar n}\hat\psi_n(r_n\hat\rho_n)^{1/2} e^{iM^{1/2}\hat z_n}&=:  \hat\Phi\\
\end{array}
 \end{equation}
yields an approximation to $\bar\Phi$; the density of state $n$ is now a constant multiple,
$r_n\ge 0$, of the one-state density $\hat\rho_n$ defined in \eqref{density_fact}, normalized to
$\int_{\tset^{3(J+N)}} \hat\rho_n dX=1$ and  $\sum_{n=1}^{\bar n}r_n=1$.
In the case of no caustics, the Ehrenfest states $(\hat X_n,\hat\psi_n,\hat\rho_n), \ n=1,\ldots ,\bar n,$ 
satisfying 
\eqref{ehrenfest_dyn} and \eqref{density_fact},
are asymptotically uncoupled, see Section \ref{sec:observables}; a caustic point $a$ is where $1/\hat\rho_n(a)=0$.
In the case of caustics and $\bar n$ colliding characteristics, the phases become coupled
and it is necessary to 
have a sum of $\bar n$ WKB terms for $\hat\Phi$ to approximate the eigenfunction $\Phi$, in \eqref{schrodinger_stat},
away from caustic points,
see \cite{keller,maslov}.  Consequently in the presence of caustics, surface-hopping approximation may
improve poor approximation  by Ehrenfest dynamics.

\section{Computation of observables}\label{sec:observables}
Assume the goal is to compute a real-valued {\it observable} 
\[\int_{\tset^{3N}}\bar\Phi\cdot A\bar\Phi dX\] for a given bounded linear multiplication operator $A=A(X)$
on $L^{2}(\tset^{3N})$ and a solution $\bar\Phi$ of \eqref{phi_ekvation}.
We have
\begin{equation}\label{observ_expand}
\begin{array}{ll}
\int_{\tset^{3N}}\bar\Phi\cdot A\bar\Phi dX
&=\sum_{n,m} \int_{\tset^{3N}}A(\psi_ne^{iM^{1/2}\theta_n})\cdot \psi_m e^{iM^{1/2}\theta_m} dX\\
&=\sum_{n,m}\int_{\tset^{3N}}Ae^{iM^{1/2}(\theta_m-\theta_n)} (\psi_n\cdot \psi_m) dX.
\end{array}
\end{equation}
The integrand is oscillatory for $n\ne m$, so that critical points of the phase difference give the main contribution:
the method of stationary phase (cf. \cite{duistermaat,maslov} and Section \ref{stationary_phase}) shows that these integrals  are small, 
bounded by $\mathcal{O}(M^{-3N/4})$,
in the case when the phase  difference has non degenerate critical points (or no critical point)
 and the functions $A \psi_n\cdot \psi_m$, $\theta_n$  are sufficiently smooth.
 A critical point $a$  satisfies $\partial_X \theta_n(a)-\partial_X\theta_m(a)=0$, 
 which means that the two different paths, generated by $\theta_n$ and $\theta_m$,
 through $a$ also have the same velocity $p$ in this point. Consequently the critical point
must be a caustic  point since  otherwise the paths are the same.
 That the critical point is non degenerate means that
 the matrix $\partial_{XX}(\theta_n-\theta_m)(a)$ is non singular.

We see that WKB terms,  with smooth phases and coefficients avoiding caustics,
 are asymptotically  orthogonal, so that the density of a linear combination of them,
separates asymptotically to a sum of densities of the individual WKB terms:
\begin{equation}\label{obser_decay}
\int_{\tset^{3N}}\bar\Phi\cdot A\bar\Phi dX
=\sum_{n=1}^{\bar n}\int_{\tset^{3N}}A\, \underbrace{\psi_n\cdot  \psi_n}_{=\rho_n\, r_n} dX +\mathcal{O}(M^{-1}),
\end{equation}
in the case of multiple eigenstates, $\bar n>1$, and 
\[
\int_{\tset^{3N}}\bar\Phi\cdot A\bar\Phi dX
=\int_{\tset^{3N}}A\, \psi_1\cdot  \psi_1 dX
\]
for a single eigenstate.
We will study molecular dynamics approximations of a single state
\begin{equation}\label{observable_sum}
\int_{\tset^{3N} }A\, \psi_n\cdot \psi_n dX
=
\int_{\tset^{3N}} A(X) \rho_n(X) dX. 
\end{equation}
in the next section.

 In the presence of a
caustic, the WKB terms can be asymptotically non orthogonal,  since their coefficients and phases typically are
not smooth enough to allow the integration by parts  to gain powers of $M^{-1/2}$.  Non orthogonal WKB functions
tells how
the caustic couples the WKB modes.

To numerically compute the integral \eqref{observable_sum} requires to find an approximation $\hat\rho_n$  of $\rho_n$,
using the Ehrenfest solution $(\hat X_n, \hat\psi_n)$, and 
to replace the integrals by quadrature (with a finite number of points $X$).
The quadrature approximation is straight forward in theory, although costly in practical computations.
Regarding the inflow density  $\hat \rho_n\big|_I$ there are two situations - either the characteristics
return often to the inflow domain or not. If they do not return we have a scattering problem and
it is reasonable to define the inflow-density  
$\hat \rho_n\big|_I$ as an initial condition. If characteristics return,
the dynamics can be used to estimate the return-density  $\hat \rho_n\big|_I$ as follows:
assume that
the following limits exist
\begin{equation}\label{ergod_limit}
\begin{array}{ll}
\lim_{T\rightarrow\infty}
\frac{1}{T}\int_0^{T} A(\hat X_t) dt 
&= \int_{\rset^{3N}}A(\hat X) \hat\rho_n(\hat X) d\hat X
\end{array}
\end{equation}
which bypasses   the need to find $\hat\rho_n\big|_I$ and the quadrature in the number of characteristics.
A way to think about this limit  is to sample the return points $\hat X_t\in I$ and from these samples construct an {\it empirical} return-density, 
converging to $\hat\rho_n\big|_I$ as the number of return iterations tends to infinity. 
We shall use this perspective to view the Eikonal equation \eqref{hj_classic} as a hitting problem
on $I$, with hitting times $\tau$ (i.e. return times).
We allow the density $\hat\rho|_I$ to depend on the initial
position $\hat X_0$; the more restrictive property  to have $\hat\rho|_I$ constant as a function
of $\hat X_0$ is called {\it ergodicity}. 

\section{Approximation of a deterministic WKB state by Hamiltonian dynamics}\label{sec_hj}
A numerical computation of an approximation to $\sum_{n}\int_{\tset^{3N} }\psi_n\cdot A\psi_n dX$ has the main ingredients:
\begin{itemize}
\item[(1)] to approximate the exact characteristics by Ehrenfest characteristics \eqref{ehrenfest_dyn},
\item[(2)] to discretize the Ehrenfest characteristics equations, and either
\item[(3)] if $\rho\big|_I$ is  an inflow-density, to introduce quadrature in the number of characteristics, or
\item[(4)] if $\rho\big|_I$ is  a return-density, to replace the ensemble average by a time average 
using the property \eqref{ergod_limit}. 
\end{itemize}
This section presents a derivation of the approximation error in step one, in the case of
a return density, avoiding the second, third and fourth
discretization steps studied
for instance in \cite{cances,lebris_hand,lebris}.

\subsection{The Ehrenfest approximation error}\label{sec:ehren_approx}
This section shows that the %two alternative 
Ehrenfest system,  written as a Hamiltonian system \eqref{psi_X_eq}, 
%or the alternative form \eqref{ehrenfest_dyn}, 
approximates Schr\"odinger observables without caustics. We see that the approximate wave function
$\hat\Phi$ defined by
\begin{equation}
\label{ehren_dynamik}
\hat\Phi=\hat\rho^{1/2}\hat\psi  e^{ iM^{1/2}\hat\theta}\ ,
\end{equation}
where $\hat\psi=\hat\phi e^{iM^{1/2}\int_0^t\hat\phi\cdot V\hat\phi(s)ds}$
and $(\hat X,\hat\phi)$ solves the Hamiltonian system \eqref{hj-system} or the system \eqref{transform_phi},
is an approximate solution to the Schr\"odinger equation \eqref{schrodinger_stat}
\begin{equation}\label{h-e_approx}
(H-E)\hat\Phi=
-\frac{1}{2M}e^{ iM^{1/2}\hat\theta} \sum_j\Delta_{X^j} (\hat\rho^{1/2}\hat\psi),
\end{equation}
since by \eqref{wkb_eq}, \eqref{tilde_psi_eq}, \eqref{hat_theta_eq} and \eqref{ehrenfest_dyn}
\begin{equation}\label{h-e-approx}
\begin{array}{ll}
(H-E)\hat\Phi&=-\underbrace{\big(iM^{-1/2}\dot{\hat\psi} - (V-\hat V_0)\hat\psi\big)}_{=0}
\hat\rho^{1/2}e^{ iM^{1/2}\hat\theta} \\
&\quad + \underbrace{\big( \frac{|\partial_X\hat\theta|^2}{2} + \hat\psi\cdot V\hat\psi -E\big)}_{=0} 
\hat\rho^{1/2}\hat\psi e^{ iM^{1/2}\hat\theta} \\
&\quad -\frac{ e^{ iM^{1/2}  \hat\theta}} {2M}\sum_j\Delta_{X^j} (\hat\rho^{1/2}\hat\psi)\\
&=-\frac{ 1} {2M} e^{ iM^{1/2}  \hat\theta}\sum_j\Delta_{X^j} (\hat\rho^{1/2}\hat\psi).
\end{array}
\end{equation}
Therefore $\hat\Phi$ approximates  a single WKB eigenstate $\Phi$, satisfying $H\Phi=E\Phi$.
The following theorem presents conditions for accurate approximation error of observables
\[
\int_{\tset^{3N}} g(X) \ \underbrace{\Phi\cdot\Phi}_{=\rho(X)} dX  
- \int_{\tset^{3N}} g(X) \underbrace{\hat\Phi \cdot \hat\Phi}_{=\hat\rho(X)}dX,
\]
using that either  the spectral gap condition \eqref{gap_c1} {\it or}
the crossing eigenvalue condition \eqref{gap_cr1} holds, based on
the electron eigenvalues $\lambda_n$, defined by $V(X)\psi_n(X)=\lambda_n(X)\psi_n(X)$ in $L^2(dx)$.

{\it The spectral gap condition.}
The electron eigenvalues $\{\lambda_n\}$
satisfy  for some positive $c$ the spectral gap condition
\begin{equation}\label{gap_c1}
\inf_{n\ne 0,\ Y\in D} |\lambda_n( Y)-\lambda_0( Y)|>c,
\end{equation}
where $D:=\{X_t \ :\ t\ge 0\}\cup\{\hat X_t \ :\ t\ge 0\}$ is the set of exact and
approximate nuclei positions.

{\it The crossing eigenvalue condition.}
Assume that 
there is a positive constant $c$ such that the sets
of eigenvalue number $n\ne 0$ and eigenvalue crossing times $\sigma$,
with maximal hitting times $\tau$ (defined in \eqref{hitting}),
\begin{equation}\label{gap_cr1}
\begin{array}{ll}
\Omega_X&=\{(n,\sigma)\ :\ \lambda_n(X_{\sigma})=\lambda_0(X_{\sigma}),  \quad 0\le \sigma\le \tau, 
n\ne 0\}\\
\Omega_{\hat X}&=\{(n,\sigma)\ :\ \lambda_n(\hat X_{\sigma})=\lambda_0(\hat X_{\sigma}),  \quad 0\le \sigma\le \tau, 
n\ne 0\}
\end{array}
\end{equation}
consists of finite many crossing times $\sigma$
with precisely one $n$ value for each crossing
time  (but possibly different $n$ at different crossing times)
and 
\[
\begin{array}{ll}
& |\frac{d}{dt}\big(\lambda_n(X_t)-\lambda_0( X_t)\big)|_{(n,t)\in \Omega_X}>c \\
& |\frac{d}{dt}\big(\lambda_n(\hat X_t)-\lambda_0( \hat X_t)\big)|_{(n,t)\in \Omega_{\hat X}}>c.\\
\end{array}
\]

\begin{theorem}\label{thm_ehrenfestapprox}
Assume that the electron eigenvalues have a spectral gap \eqref{gap_c1} and $\alpha=1-\delta$,
or that they form non degenerate critical points \eqref{gap_cr1} and $\alpha=3/4$,
then
Ehrenfest dynamics \eqref{ehrenfest_ham},
%\eqref{ehren_dynamik}, 
assumed to avoid caustics
and have bounded hitting times $\tau$ in \eqref{hitting}, approximates time-independent
Schr\"odinger observables with the error bounded by $\mathcal{O}(M^{-\alpha}):$
\begin{equation}\label{hj_estimate}
\begin{array}{ll}
\int_{\tset^{3N}}g(\hat X)\hat\rho(\hat X) d\hat X
&=\int_{\tset^{3N}}g(X)\rho( X) d X + \mathcal{O}(M^{-\alpha}), \mbox{ for any } \delta> 0.
\end{array}
\end{equation}
\end{theorem}
The proof is in Section \ref{born_oppen}, with the main steps
separated into  six subsections: spectral decomposition,
Schr\"odinger eigenvalues as a Hamiltonian system, stability from perturbed Hamiltonians,
the Born-Oppenheimer approximation, the stationary phase method, and error estimates of the densities.
 The observable may include the time variable, through the position of a
non interacting (non-charged) nucleus  moving with given velocity, so
that transport properties as e.g. the diffusion and viscosity of a liquid may
by determined, cf. \cite{frenkel}.
A case with no caustics is e.g. when $\min_X(E-V_0(X))>0$ in one dimension $X\in \rset$.

\subsection{The Born-Oppenheimer approximation error}\label{BO_stycke}
The Born-Oppenheimer approximation leads to the 
standard formulation of {\it ab initio} molecular dynamics, in the micro-canonical ensemble
with constant number of particles, volume and energy, for the nuclei positions  $\bar X$,
\begin{equation}\label{bo_stated}
\begin{array}{ll}
\dot{\bar X}_t &= \bar p_t\\
\dot{\bar p}_t&=- \partial_X \lambda_0(\bar X_t)\\
%V_0(X) &:= \frac{\psi_0\cdot V(X)\psi_0}{\psi_0\cdot\psi_0},
\end{array}
\end{equation}
by using that the electrons are in the eigenstate  $\psi_0$ with eigenvalue $\lambda_0$ to $V$, in $L^2(dx)$
for fixed $X$. The corresponding Hamiltonian is $H_{BO}(\bar X,\bar p):= |p|^2/2 + \lambda_0(X)$.
The following theorem shows that
Born-Oppenheimer dynamics approximates  Schr\"odinger observables, based on a single WKB eigenstate,
as accurate as the Ehrenfest dynamics, using the approximate Schr\"odinger solution
 \[
\hat\Phi=\bar G^{-1}\psi_0  e^{ iM^{1/2}\bar\theta_0},\] 
with density $\bar\rho:=\hat\Phi\cdot\hat\Phi$ and $\bar G$ defined by $d\log \bar G/dt = \mbox{div} \bar p/2$,
as for $G$ in \eqref{s_int}
\begin{theorem}\label{bo_thm}
Assume that the electron eigenvalues have a spectral gap \eqref{gap_c1} and $\alpha=1$,
or that they form non degenerate critical points \eqref{gap_cr1} and $\alpha=1/2$,
then the zero-order Born-Oppenheimer dynamics \eqref{bo_stated}, assumed to avoid caustics
and bounded hitting times $\tau$ in \eqref{hitting}, approximates time-independent
Schr\"odinger observables with error bounded by $\mathcal{O}(M^{-\alpha+\delta}):$
\begin{equation}\label{bo_estimate}
\begin{array}{ll}
\int_{\tset^{3N}}g(\bar X)\bar\rho( \bar X) d \bar X
&=\int_{\tset^{3N}}g(X)\rho( X) d X + \mathcal{O}(M^{-\alpha+\delta}), \mbox{ for any } \delta>0.
\end{array}
\end{equation}
\end{theorem}

The proof is in Section \ref{born_oppen}. 

\begin{remark}[Why do symplectic numerical simulations of molecular dynamics work?]\label{md_sim}
The derivation of the approximation error for the Ehrenfest  and Born-Oppenheimer dynamics, 
in Theorems \ref{thm_ehrenfestapprox} and \ref{bo_thm},
also allows
to study perturbed systems.  For instance, the perturbed 
Born-Oppenheimer dynamics 
\[
\begin{array}{ll}
\dot X&=p + \partial_pH^{\delta}(X,p,\hat\psi)\\
\dot p&=-\partial_X\lambda_0(X) -\partial_X H^\delta(X,p,\hat\psi),\\
\end{array}
\]
generated
from a perturbed Hamiltonian $H_{BO}(X,p,\hat\psi)+ H^\delta(X,p,\hat\psi)=E$,
with the perturbation satisfying
\begin{equation}\label{stab_h_d}
\|H^\delta\|_{L^\infty} \le \delta
\quad \mbox{ for some $\delta\in (0,\infty)$}
\end{equation}
yields through \eqref{H_felet} and \eqref{g_fel} an additional error term $\mathcal{O}(\delta)$
to the approximation of observables in \eqref{hj_estimate}. So called symplectic numerical methods are precisely those that can
be written as perturbed Hamiltonian systems, see \cite{ms}, and consequently we have a method to precisely analyze their numerical error
by combining an explicit construction of $H^\delta$ 
with the stability condition \eqref{stab_h_d}
to obtain $\mathcal{O}(M^{-1}+\delta)$ accurate approximations.
The popular St\"ormer-Verlet method  is symplectic and the positions $X$ coincides with those of
the symplectic Euler method,  for which $H^\delta$ is explicitly constructed in \cite{ms}
with  $\delta$ proportional to the time step. \end{remark}

\section{Analysis of the molecular dynamics approximation}\label{born_oppen}

This section continues the construction of the solution operator started in Section \ref{start_sec}.
It is written for the Schr\"odinger characteristics, but it can be directly applied to the Ehrenfest
case by replacing $\tilde V$ by $V-\hat V_0$ (by formally taking the limit $M\rightarrow\infty$ in the term
$(2M)^{-1}G\sum_j\Delta X_j (G^{-1}\tilde\psi)$).
Assume for a moment that $\tilde V$ is  independent of $\tau$. 
Then the solution to \eqref{new_psi_eq}
can be written
as a linear combination of the two exponentials
\[
Ae^{i\tau\alpha_+} + Be^{i\tau \alpha_-}
\]
where the two characteristic roots are
\[
\alpha_\pm=(p_1^1)^2\big(-1\pm (1-2(p_1^1)^{-2}\tilde V)^{1/2}\big).
\]
We see that $e^{i\tau\alpha_-}$ is a highly oscillatory solution on the fast $\tau$-scale with 
\[
\alpha_-=-2(p_1^1)^2 +\tilde V+ \mathcal{O}\big(\tilde V^2/(p_1^1)^2\big),\]
while 
\begin{equation}\label{alfa_1_bound}
\alpha_+= -\tilde V +\mathcal{O}(\tilde V^2/(p_1^1)^2).
\end{equation}
Therefore we chose initial data 
\begin{equation}\label{tau_const}
i\dot{\tilde\psi}|_{\tau=0}=-\alpha_+\tilde\psi|_{\tau=0}
\end{equation}
to have $B=0$, which eliminates the fast scale, and  the
 limit $p_1^1\rightarrow\infty$  determines the solution by
the Schr\"odinger equation
\[
i\dot{ \tilde\psi}=\tilde V\tilde\psi.
\]
In the case of the Ehrenfest dynamics, this equation with $\tilde V$ replaced by $V-\hat V_0$ is the starting point.
The next section presents an analogous  construction for the
slowly, in $\tau$, varying  operator $\tilde V$.

\subsection{Spectral decomposition}
Write \eqref{new_psi_eq} as the first order system
\[
\begin{array}{ll}
i\dot{\tilde\psi} &=v\\
\dot v&=2(p_1^1)^2 i(\tilde V\tilde\psi-v)
\end{array}
\]
which for $\bar\psi:=(\tilde\psi,v)$ takes the form
\[
\dot{\bar\psi}=iA\bar\psi\quad A:=
\left(\begin{array}{cc}
0 & -1\\
2(p_1^1)^2\tilde V & -2(p_1^1)^2\\
\end{array}\right),
\]
where the  eigenvalues $\lambda_\pm$ , right eigenvectors $q_\pm$ and left eigenvectors $q^{-1}_\pm$ of the real "matrix" operator $A$ are
\[
\begin{array}{ll}
\lambda_\pm &:= (p_1^1)^2\Big(-1\pm \big(1-2 (p_1^1)^{-2}\tilde V\big)^{1/2}\Big),\\
q_+ &:=\left(\begin{array}{c}
1\\
-\lambda_+\\
\end{array}\right),\\
q_- &:=\left(\begin{array}{c}
-\lambda_-^{-1}\\
1\\
\end{array}\right),\\
q_+^{-1} &:=\frac{1}{1-\lambda_+/\lambda_ -}\left(\begin{array}{c}
1\\
\lambda_-^{-1}\\
\end{array}\right),\\
q_-^{-1} &:=\frac{1}{1-\lambda_+/\lambda_-}\left(\begin{array}{c}
\lambda_+\\
1\\
\end{array}\right).
\end{array}
\]
We see that $\lambda_+=-\tilde V+\mathcal{O}\big((p_1^1)^{-2}\big)$ and 
$\lambda_-=-2(p_1^1)^{2} +\tilde V+\mathcal{O}\big((p_1^1)^{-2}\big)$.
The important property here is that the left eigenvector limit
$\lim_{p_1^1\rightarrow\infty} q_+^{-1}=(1,0) $ is constant, independent of $\tau$,
which implies that the $q_+$ component $ q_+^{-1}\bar\psi=\tilde\psi$ decouples:
we obtain in the limit $p_1^1\rightarrow\infty$ the time-dependent Schr\"odinger equation
\[
\begin{array}{ll}
i\dot{\tilde\psi}(\tau)&=i \frac{d}{d\tau}( q_+^{-1}{\bar\psi}_\tau) \\
&= i q_+^{-1}\frac{d}{d\tau}{\bar\psi}_\tau \\&
=- q_+^{-1}A_\tau \bar\psi_\tau\\
&=-\lambda_+(\tau) q_+^{-1} \bar\psi_\tau\\
&=-\lambda_+(\tau) \tilde\psi(\tau) \\ 
&=\tilde V_\tau \tilde\psi(\tau), \\ 
\end{array}
\]
where the operator $\tilde V_\tau$ depends on $
\tau$ and $(x,X_0)$, and we define
the solution operator $S$
\begin{equation}\label{psi_evolution}
\tilde\psi(\tau)=S_{\tau,0}\tilde \psi(0).
\end{equation}
As in \eqref{tau_const} we can view this as choosing special initial data for $\tilde\psi(0)$;
from now on we only consider such data.

The operator $\tilde V$ can be symmetrized 
\begin{equation}\label{v_sym}
\bar V_\tau:=G_{\tau}^{-1} \tilde V_\tau G_{\tau}
= (V-V_0)_\tau - \frac{1}{2M} \sum_{j} \Delta_{X^j_0},
\end{equation}
with real eigenvalues  $\{\check\lambda_m\}$ and orthonormal eigenvectors $\{p_m\}$ in $L^2(dxdX_0)$,
satisfying 
\[
\bar V_\tau p_m(\tau)=\check\lambda_m(\tau) p_m(\tau),
\]
see Section \ref{fredholm}.
Therefore $\tilde V_\tau$ has the same eigenvalues and the eigenvectors
$\bar p_m:=G_\tau p_m$, which establishes  the spectral representation
\begin{equation}\label{spectrum}
\tilde V_\tau \tilde\psi(\cdot,\tau,\cdot)=\sum_m \check\lambda_m(\tau)
\int_{\tset^{3N-1} }\tilde\psi(\cdot,0,\cdot) \cdot \bar p_m  G_\tau^{-2} dX_0
%\tilde\psi(\cdot,0,\cdot)\bullet \bar p_m(\tau) 
\ \bar p_m(\tau),
\end{equation}
where the
scalar product is
\begin{equation}\label{g_vikt}
%\tilde\psi \bullet \bar p_m :=
\int_{\tset^{3N-1} }\tilde\psi \cdot \bar p_m  G_\tau^{-2} dX_0.
\end{equation}
We note that the weight $G^{-2}$ on the co-dimension one surface $\tset^{3N-1}$
appears  precisely because the operator $\tilde V$ is symmetrized by $G$
and the weight $G^{-2}$ corresponds to the Eulerian-Lagrangian change of coordinates \eqref{euler_lagrange-det}.
The existence of the orthonormal set of eigenvectors and real eigenvalues
makes the operator $\tilde V$ essentially self-adjoint 
in the Lagrangian coordinates and hence the solution operator $S$
becomes unitary in the Lagrangian coordinates.
In the case of the Ehrenfest dynamics the weight is the density $\hat\rho=G^{-2}$;

\subsection{Schr\"odinger eigenvalues as a Hamiltonian system}\label{sec:s-ham}
Our error analysis is based on comparing perturbations of Hamiltonian systems.
This section establishes
the Schr\"odinger eigenvalue characteristics as a Hamiltonian system,
analogously to the Hamiltonian system \eqref{ehrenfest_ham} for Ehrenfest dynamics.

To write the Schr\"odinger eigenvalue characteristics  as a Hamiltonian system
requires to make an asymptotically negligible change in the definition of $V_0$ in \eqref{psi_first_eq}:
the  Eikonal equation used the leading order term 
$V_0=\tilde\psi\cdot V\tilde\psi/\tilde\psi\cdot\tilde\psi$.
An alternative choice is to replace $V$ by  the complete Schr\"odinger operator 
\[
\check V\tilde\psi:=V\tilde\psi -G(2M)^{-1}\sum_n\Delta_{X_n} (\tilde\psi G^{-1})
\]
and symmetrize 
\begin{equation}\label{v_0-def}
V_0:= \frac{\tilde\psi\cdot \check V\tilde\psi 
+\check V\tilde\psi\cdot\tilde\psi}{2\, \int \tilde\psi\cdot\tilde\psi \, G^{-2} dX/\int G^{-2} dX}. 
\end{equation}
The integral in $X$ is used in the denominator
to avoid the  non constant $\tilde\psi\cdot\tilde\psi$;
since $\check V$ is essentially self-adjoint in the weighted space $L^2(G^{-2}dX)$ in \eqref{g_vikt},
corresponding to Lagrangian coordinates, the weighted $L^2$ norm in the denominator
is constant in time and consequently the denominator can be treated as a constant
when differentiating the Hamiltonian to obtain the Hamiltonian system. 

This definition of $V_0$ has the same leading order term as the previous \eqref{v0-def}.
Equation \eqref{psi_first_eq} showed that 
\[
\begin{array}{ll}
0&=(H-E)\Phi\\
&=\Big(\big(-\frac{i}{M^{1/2}}\dot{ \tilde\psi}+ (V-V_0)\tilde\psi 
-\frac{1}{2M}\sum_j G\Delta_{X^j}(\tilde\psi G^{-1})\big) G^{-1}\\
&\quad
+{\big( \frac{|\partial_X\theta|^2}{2} + V_0 -E\big)} \tilde\psi\Big)e^{iM^{1/2}\theta(X)}\, ,
\end{array}
\] for any definition of the scalar $V_0$, which we now use for the new $V_0$.
Note that, it is only when $(V-V_0)\tilde\psi$ is  small that we expect 
the $X$-derivative of $\tilde\psi$ to be bounded,
since otherwise the $e^{iM^{1/2}\theta}$ scale will pollute derivatives: $p\cdot \partial_X\tilde\psi=-iM^{1/2}(\check V-V_0)\tilde\psi$. Therefore, the interesting case is
when $\tilde\psi$ is close to an electron eigenfunction and $V_0$ is near the corresponding eigenvalue.

The goal here is to show that the Eikonal equation, corresponding to the second term above,
becomes a Hamiltonian system for the value function $\theta(X,p,\varphi)$, so that it also generates the 
Schr\"odinger equation in the first term above, analogously to the Ehrenfest dynamics \eqref{ehrenfest_ham}.
Define $\varphi=2M^{1/4}\phi$ and the Hamiltonian
\begin{equation}\label{S_hamiltonian}
H_S(X,p,\varphi) :=\frac{|p|^2}{2} +\frac{\phi\cdot \check V\phi +\check V\phi\cdot \phi}{2} ,
\end{equation}
similar as in \eqref{ehrenfest_ham},
which yields the Hamiltonian system
\begin{equation}\label{s-hamilton}
\begin{array}{ll}
\dot X &=p\\
\dot p&= -\frac{\phi\cdot \partial_X\check V\phi +\partial_X\check V\phi\cdot \phi}{2}\\
i\dot \phi &= M^{1/2}\check V\phi
\end{array}
\end{equation}
using that $\check V$ is symmetric in the weighted scalar product $L^2(dxG^{-2} dX)$, see \eqref{v_sym}.
The phase factor change 
$\tilde\psi:=\phi e^{ iM^{1/2}\int_0^t 
({\phi\cdot \check V\phi +\check V\phi\cdot\phi})/{(2\, \int \phi\cdot\phi \, G^{-2} dX/\int G^{-2} dX) ds
%2\, \hat\phi\cdot\hat\phi
}}$
implies that
\[
\begin{array}{ll}
\dot X &=p\\
\dot p&= -\frac{\tilde\psi\cdot \partial_X\check V\tilde\psi 
+\partial_X\check V\tilde\psi\cdot \tilde\psi}{2\, \int \tilde\psi\cdot\tilde\psi \, G^{-2} dX/\int G^{-2} dX}\\
i\dot {\tilde\psi} &= M^{1/2}\underbrace{(\check V-V_0)}_{=:\tilde V}\tilde\psi
\end{array}
\]
for $V_0$ defined by \eqref{v_0-def}
and 
\[
|\partial_X\theta|^2/2 + V_0=E.
\]
The last two equations imply that $(X,p,\tilde\psi)$ and the weight $G$, defined by \eqref{s_int}
(using the new $V_0$),
generate a WKB function $\Phi:= G^{-1}\tilde\psi e^{iM^{1/2}\theta}$  that solves the Schr\"odinger
equation $H\Phi=E\Phi$ in \eqref{schrodinger_stat}. 
Therefore, the Hamiltonian \eqref{S_hamiltonian} generates an exact solution to the Schr\"odinger equation.

 \subsection{Stability from perturbed Hamiltonians}\label{pert_ham}
This section derives error estimates
of the weight functions $G$
when the corresponding Hamiltonian system is perturbed.

To derive the stability estimate
we consider the 
Hamilton-Jacobi equation  \[H(\partial_Y\theta(Y),Y)=0\] in an optimal control
perspective, with the corresponding Hamiltonian system
\[
\begin{array}{ll}
\dot Y_t &= \partial_{q} H(q_t,Y_t)   \\ 
\dot q_t &= -\partial_{Y} H(q_t,Y_t). 
\end{array}
\]
Define  the "value" function
\[
\theta(Y_0)= \theta(Y_t) - \int_0^t h(q_s,Y_s) ds
\]
where the "cost" function defined by
\[
h(q,Y):=  q\bullet \partial_q H(q,Y)-H(q,Y)
\]
satisfies the Pontryagin principle (related to the Legendre transform)
\begin{equation}\label{pontry}
H(q,Y)= \sup_Q \big( q\bullet \partial_Q H(Q,Y) - h(Q,Y)\big).
\end{equation}
Let $\theta\Big|_I$ be defined by the hitting problem
\[
\theta(Y_0 )=\theta(Y_\tau) -\int_0^\tau  h(q_s,Y_s) ds
\]
using the hitting time $\tau$ on the return surface $I$
\begin{equation}\label{hitting}
\tau:= \inf\{ t \ : \ Y_t\in I\ \& \ t>0\}.
\end{equation}
Define analogously for a perturbed Hamiltonian $\tilde H$,
the dynamics $(\tilde Y_t,\tilde q_t)$, the value function $\tilde\theta$ and
the cost function $\tilde h$.

We can think of the difference $\theta-\tilde\theta$ as composed by perturbation of the boundary
data (on the return surface $I$) and perturbations of the Hamiltonians. 
The difference  of the value functions due to the perturbed Hamiltonian satisfy the stability estimate
\begin{equation}\label{theta_alfa}
\begin{array}{ll} 
\theta(Y_0)-\tilde\theta(Y_0) &\ge  \theta(\tilde Y_{\tilde\tau})-\tilde\theta(\tilde Y_{\tilde\tau}) +
\int_0^{\tilde\tau} (H-\tilde H)\big(\partial_Y\theta(\tilde Y_t),\tilde Y_t\big) dt \\
\theta(Y_0)-\tilde\theta(Y_0) &\le \theta( Y_{\tau})-\tilde\theta( Y_{\tau}) +
\int_0^{\tau} (H-\tilde H)\big(\partial_Y\tilde\theta(Y_t), Y_t\big) dt 
\end{array}
\end{equation}
with a difference of the Hamiltonians evaluated along the same solution path.
This result follows by differentiating the value function along a path and
using the Hamilton-Jacobi equations, see Remark \ref{hj_stab} and \cite{css}.

Assume that
\begin{equation}\label{H_felet}
\sup_{(q,Y)=(\partial_Y\tilde\theta(Y_t), Y_t),\  (q,Y)=(\partial_Y\theta(\tilde Y_t), \tilde Y_t)}
 |(H-\tilde H)(q,Y)|=\mathcal O(M^{-\alpha}).
\end{equation}
The stability estimate, for all characteristic paths and following subsequent
hitting points on $I$ (as in the discussion after \eqref{char_eq}), then yields the same  estimate of the difference in the boundary data
\[
 \|\theta-\tilde\theta\|_{L^\infty(I)}=\mathcal O(M^{-\alpha})
\]
provided the maximal hitting time $\tau$ is bounded,
which we assume. If the return surface is the plane for which
the $X_1$ particle has its first component equal to its initial value, 
the hitting time is the time it takes until this particle has the first
component equal to that initial value again; since the $X$-dynamics
does not explicitly depend on $M$ it seems reasonable that one
can find a return surface such that the hitting times are bounded.
Next, the representation can be applied to interior points
to obtain 
\[
\|\theta-\tilde\theta\|_{L^\infty} = \mathcal O(M^{-\alpha}). 
\] 

When the value functions $\theta$ and $\tilde\theta$ are smoothly differentiable in $X$
($X$ is a part of the $Y$ coordinate) with 
derivatives bounded uniformly in $M$, 
the stability estimate \eqref{theta_alfa} implies that also the difference of the second derivatives
has the bound
\[
\|\sum_n \Delta_{X_n}\theta -\sum_n \Delta_{X_n}\tilde\theta \|_{L^\infty}=\mathcal{O}(M^{-\alpha+\delta}), \mbox{ for any } \delta>0.
\]
We will also use that
\begin{equation}\label{imag_alfa}
M^{-1/2} \Im (\psi\cdot\sum_n\Delta_{X_n}\psi)/(\psi\cdot\psi)=\mathcal{O}(M^{-\alpha})
\end{equation}
which is verified in Section \ref{imag_rem}.

Our goal is to analyze the  density function $\rho$ satisfying the convection equation
\eqref{s_int}, i.e.
\begin{equation}\label{g_hj}
\partial_X\theta\bullet \partial_X \log \rho =- \sum_n \Delta_{X_n}\theta + d
\end{equation}
where the function $d=d_{BO}, d= d_E$ or $d=d_S$ with
\[
\begin{split}
d_{BO} &=d_{E}=0\\
d_{S} &=  -M^{-1/2} \Im (\psi\cdot\sum_n\Delta_{X_n}\psi)/(\psi\cdot\psi)
\end{split}
\]
in the case of Born-Oppenheimer, Ehrenfest or Schr\"odinger dynamics, respectively.
%
%weight function $G:=e^{G_d}$ satisfying the convection equation
%\eqref{s_int}, i.e.
%\begin{equation}\label{g_hj}
%\partial_X\theta\bullet \partial_XG_d =\frac{1}{2}\sum_n \Delta_{X_n}\theta
%\end{equation}
%where $\partial_X\theta\bullet \partial_XG_d=dG_d(X_t)/dt$. 
We have $\partial_X\theta\bullet \partial_X \log\rho =d\log\rho(X_t)/dt$ and
we interpret equation \eqref{g_hj} as a Hamilton-Jacobi equation
and similarly for $\log\tilde \rho$. Then the stability of Hamilton-Jacobi equations, as in \eqref{theta_alfa},
can be applied to \eqref{g_hj}, with the Hamiltonian 
%\[
%H_G(\partial_X G_d,X):=\partial_X\theta(X)\bullet \partial_XG_d -\frac{1}{2}\sum_n \Delta_{X_n}\theta(X),\] 
%to obtain the sought error estimate
%\begin{equation}\label{g_fel}
%\|G-\tilde G\|_{L^\infty} 
%=\mathcal O(M^{-\alpha}).
%\end{equation}
\[
H_\rho (\partial_X \log\rho,X):=\partial_X\theta(X)\bullet \partial_X\log\rho(X)  + \sum_n \Delta_{X_n}\theta(X) - d,\] 
to obtain the sought error estimate
\begin{equation}\label{g_fel}
\|\log\rho-\log\tilde\rho\|_{L^\infty} 
=\mathcal O(M^{-\alpha+\delta}).
\end{equation}
In this sense we will use that an $\mathcal O(M^{-\alpha})$ perturbation of the Hamiltonian yields an error
estimate of almost the same order for the difference of the corresponding %weight functions $G-\tilde G$.
densities $\rho-\tilde\rho$.

The Hamiltonians we use are
\[\begin{array}{ll}
H_S& = \frac{|p|^2}{2} + \frac{1}{2} \big( \phi\cdot \check V(X)\phi 
+ \check V(X)\phi\cdot \phi\big) -E \\
H_E & = \frac{|p|^2}{2} +  \phi\cdot V(X)\phi -E \\
H_{BO} & = \frac{|p|^2}{2} +  \lambda_0(X) -E
\end{array}\]
with the cost functions
\[\begin{array}{ll}
h_S& = E + \frac{|p|^2}{2}  + \frac{1}{2} \big( \phi_i\cdot \check V(X)\phi_i 
+ \check V(X)\phi_i\cdot \phi_i  -\phi_r\cdot \check V(X)\phi_r 
-\check V(X)\phi_r\cdot \phi_r \big) \\
h_E& = E + \frac{|p|^2}{2}  + %\frac{1}{2} 
\big( \phi_i\cdot V(X)\phi_i 
-\phi_r\cdot V(X)\phi_r \big) \\
h_{BO}& = E + \frac{|p|^2}{2} - \lambda_0(X)
\end{array}\]
and the primal and dual variable $(Y;q)= (X,2^{1/2}M^{-1/4}\phi_r; p,2^{1/2}M^{-1/4}\phi_i)$ in the case of
Schr\"odinger and Ehrenfest dynamics, as in \eqref{E-hamiltonian},
and the variable $(Y;q)=(X,(\psi_0)_r;p,(\psi_0)_i)$ for the Born-Oppenheimer dynamics.
For the Born-Oppenheimer case the electron wave function is the eigenstate $\psi_0$;
%there are two different possibilities of comparing the variables $(X;p)$ in the Born-Oppenheimer dynamics
%and $(Y;q)$ in the Ehrenfest and Schr\"odinger dynamics: 
one can identify 
the Born-Oppenheimer dynamics with an Ehrenfest dynamics where the mass is set to infinity,
since in this limit the Ehrenfest wave function becomes the eigenfunction, see Section \ref{sec:bo}.

\begin{remark}\label{hj_stab}
This remark derives the stability estimate \eqref{theta_alfa}.
The definition of the value functions imply
\begin{equation}
  \label{eq:err2}
  \begin{array}{ll}
    &\underbrace{\tilde \theta(\tilde Y_{\tilde\tau}) 
    -\int_0^{\tilde\tau} \tilde h(\tilde q_t,\tilde Y_t)\,{\rm d}t 
    }_{\tilde\theta(\tilde Y_0)} - 
    \underbrace{ \big(\theta( Y_\tau) -\int_0^\tau h( q_t, Y_t)\,{\rm d}t \big) }_{\theta(Y_0)}\\
    &=
    -\int_0^{\tilde\tau} \tilde h(\tilde q_t,\tilde Y_t) \,{\rm d}t  +\theta(\tilde Y_{\tilde \tau})
    -\underbrace{ \theta(Y_0)}_{ \theta(\tilde Y_0)} 
     +  \tilde\theta(\tilde Y_{\tilde \tau}) - \theta(\tilde Y_{\tilde \tau})\\
    &=  -\int_0^{\tilde\tau} \tilde h(\tilde q_t,\tilde Y_t) \,{\rm d}t  
    +  \int_0^{\tilde\tau} {\rm d}\theta(\tilde Y_t)
     +  \tilde\theta(\tilde Y_{\tilde \tau}) - \theta(\tilde Y_{\tilde \tau})\\
    &=\int_0^{\tilde\tau}\underbrace{ -\tilde h(\tilde q_t,\tilde Y_t) +
      \partial_Y \theta (\tilde Y_t) \bullet 
      \partial_q \tilde H(\tilde q_t,\tilde Y_t) }_{\le 
      \tilde H\big(\partial_Y \theta(\tilde Y_t), \tilde Y_t\big) }\,{\rm  d}t
       +  \tilde\theta(\tilde Y_{\tilde \tau}) - \theta(\tilde Y_{\tilde \tau})\\
    &\le \int_0^{\tilde\tau} (\tilde H-H)\big(\partial_Y \theta(\tilde Y_t), \tilde Y_t\big)\,{\rm  d}t
     +  \tilde\theta(\tilde Y_{\tilde \tau}) - \theta(\tilde Y_{\tilde \tau}),
  \end{array}
\end{equation}
where the Pontryagin principle \eqref{pontry} yields the inequality and
we use the Hamilton-Jacobi equation \[H(\partial_Y \theta(\tilde Y_t), \tilde Y_t)=0.\]
To establish the lower bound, replace $\theta$ along $\tilde Y_t$ by
$\tilde\theta$ along $Y_t$ and repeat the derivation above.
\end{remark}

\subsection{The Born-Oppenheimer approximation}\label{sec:bo}
\subsubsection{With an electron eigenvalue gap}
To better understand the evolution \eqref{ehrenfest_dyn} of  $\hat\psi$
and estimate the error term in \eqref{h-e-approx}, we use the decomposition
$\hat\psi=\psi_0 \oplus \psi^\bot$, where $\psi_0(t)$ is an electron eigenvector of $V_t=V(\hat X_t)$,
satisfying $V_t\psi_0(t)=\lambda_0(t)\psi_0(t)$ in $L^2(dx)$ for an eigenvalue $\lambda_0(t)\in\rset$,
and $\psi^\perp$ is chosen orthogonal to $\psi_0$ in the sense $\psi^\perp\cdot\psi_0=0$.
We assume that the electron eigenfunction $\psi_0(X):\tset^{3J}\rightarrow \rset$
is smooth as a function of $X$.
This {\it Ansatz} is motivated by the residual
\begin{equation}\label{R_residual}
R\psi_0:= \dot\psi_0 + i M^{1/2}(V-V_0)\psi_0
\end{equation}
being small. Below we verify that  $\psi^\bot$ is $\mathcal O(M^{-1/2})$ with a spectral gap assumption
and $\mathcal O(M^{-1/4})$ for a case without a spectral gap; 
in the case of a spectral gap, this estimate yields
\[
\begin{array}{ll}
\dot\psi_0 \cdot \psi_0& =-\dot \psi^\bot\cdot \psi^\bot = \mathcal{O}(M^{-1/2})\\
M^{1/2}(V-V_0)\psi_0 
&=\mathcal{O}(M^{-1/2}),
\end{array}
\]
 since
 \[
 V_0=\frac{(\psi_0+\psi^\bot)\cdot V_\tau(\psi_0+\psi^\bot)}{(\psi_0+\psi^\bot)\cdot (\psi_0+\psi^\bot)}
 =\lambda_0(\tau)+\mathcal{O}(M^{-1})
 \]
 for $\psi_0=\mathcal{O}(1)$,
 and \[
 \dot \psi^\bot\cdot \psi^\bot = \dot X\bullet\partial_X\psi^\bot\cdot\psi^\bot= \mathcal{O}(M^{0} M^{-1/2}).
 \]

 The Schr\"odinger equation   $R\hat\psi=0$  in \eqref{ehrenfest_dyn} implies that
the perturbation $\psi^\bot$ satisfies the following Schr\"odinger equation with the source term $iR\psi_0$
\begin{equation}\label{psi_bot-ekv}
i\dot{\psi}^\bot=M^{1/2}(V-V_0)\psi^\bot - iR\psi_0
\end{equation}
and the solution representation 
\begin{equation}\label{r_int}
\psi^\bot(t)=\hat S_{t,0}\psi^\bot(0)
-\int_0^t \hat S_{t,s} R\psi_0(s)
ds
\end{equation}
for the solution operator $\hat\psi_t=:\hat S_{t,0}\hat\psi_0$.
Split the source term into its projection on $\psi_0$ and its orthogonal part
\[
R\psi_0=\frac{R\psi_0\cdot \psi_0}{\psi_0\cdot\psi_0}\ \psi_0  \oplus R\psi_0^\bot 
\]
and note that it is enough to study the projected source term $R\psi_0^\bot$ to determine $\psi^\bot$ on
the orthogonal complement of $\psi_0$.
Integrate by parts to obtain
\begin{equation}\label{s_part}
\begin{array}{ll}
\int_0^t \hat S_{t,s} R\psi_0^\bot(s) ds &= 
\int_0^t \underbrace{\frac{d}{ds} (\hat S_{t,s})iM^{-1/2}(V-V_0)^{-1}}_{=\hat S_{t,s}}  R\psi_0^\bot(s) ds \\
&=iM^{-1/2}\Big[\hat S_{t,s}(V-V_0)^{-1}  R\psi_0^\bot(s)\Big]_{s=0}^t \\
&\quad -
iM^{-1/2}\int_0^t  \hat S_{t,s}\frac{d}{ds}\big((V-V_0)^{-1}  R\psi_0^\bot\big)(s) ds.
\end{array}
\end{equation}
If there is a {\it spectral gap }
\begin{equation}\label{gap_c}
\inf_{n\ne 0, t>0} |\lambda_n(\hat X_t)-\lambda_0(\hat X_t)|>c,
\end{equation}
for a positive constant $c$,
we have
\[
(V-V_0)^{-1}R\psi_0^\bot=\sum_{n\ge 0}(V-V_0)^{-1}\psi_n \ R\psi_0^\bot\cdot \psi_n  
=\sum_{n\ne 0}(\lambda_n-V_0)^{-1} \psi_n \ R\psi_0^\bot\cdot \psi_n\ 
=\mathcal O(1).
\]
We have by \eqref{r_int} and \eqref{s_part}
\[
\begin{split}
\psi^\perp(t) &= \hat S_{t,0}\big(\psi^\perp(0) -iM^{-1/2}(V-V_0)^{-1}R\psi_0^\perp(0)\big)
+ iM^{-1/2}(V-V_0)^{-1}R\psi_0^\perp(t)\\
& \quad 
- \underbrace{iM^{-1/2}\int_0^t  \hat S_{t,s}\frac{d}{ds}\big((V-V_0)^{-1}  R\psi_0^\bot\big)(s) ds}_{=\mathcal{O}(M^{-1})},
\end{split}
\]
where the last integral can be integrated by parts again as in \eqref{s_part} to reduce the integral with a factor
$M^{-1/2}$.
By choosing $\psi^\perp(0)=iM^{-1/2}(V-V_0)^{-1}R\psi_0^\perp(0)$ % +\mathcal{O}(M^{-1})$ 
we have
\begin{equation}\label{psi_perp_ekv}
\psi^\perp(t)=iM^{-1/2}(V-V_0)^{-1}R\psi_0^\perp(t) + 
\underbrace{iM^{-1/2}\int_0^t  \hat S_{t,s}\frac{d}{ds}\big((V-V_0)^{-1}  R\psi_0^\bot\big)(s) ds}_{\mathcal{O}(M^{-1})}. %\mathcal{O}(M^{-1}),
\end{equation}
We  may use that the
paths $\hat X_t$ return to the co-dimension one "return" surface $I$, defined in \eqref{char_eq},
after time $t=\mathcal{O}(1)$ which by \eqref{s_part} then yields
\begin{equation}\label{sr_int}
\|\int_0^t \hat S_{t,s} R\psi_0^\bot(s) ds\|_{L^2}=\mathcal O(M^{-1/2}),
\end{equation}
and if longer return times are needed the integration by parts can be iterated $k$ times to gain
a factor of $M^{-k/2}$ in the integral. 
The function $\psi^\perp$ is then determined by \eqref{r_int} from the initial $\hat X\in I$ and its successive returns 
to the set $I$, similarly as the phase function $\theta$ in \eqref{char_eq}.
Therefore, the $L^2$ unitarity of $\hat S_{t,0}$  together with the bound \eqref{sr_int} imply in \eqref{r_int} that
the Ehrenfest approximation $\hat\psi=\psi_0+\hat{\psi}^\perp$
satisfies
\begin{equation}\label{zeta_uppskatt}
\|{\psi}^\perp\|_{L^2}=\mathcal{O}(M^{-1/2}),
\end{equation}
where \eqref{psi_perp_ekv} shows that the error in different hitting time intervals  do not substantially 
accumulate in time.
Since the function $\psi_0$ depends smoothly on $X$ and a derivative of the
integral in \eqref{psi_perp_ekv} yields at most a factor $M^{1/2}$ (when differentiating the solution operator $\hat S$),
equation \eqref{psi_perp_ekv} shows that the derivative with respect to the Lagrange coordinates
satisfies
\[
\|\partial \psi^\bot\|_{L^2} = \mathcal{O}(M^{-1/2}),
\]
and integrating by parts once more in \eqref{psi_perp_ekv} also gives
\begin{equation}\label{bot_bound}
\|\partial^2\psi^\bot\|_{L^2} = \mathcal{O}(M^{-1/2}).
\end{equation}

\subsubsection{With crossing electron eigenvalues}
Since the solution operator is unitary it can be written as a highly oscillatory
exponential:  $\hat S_{t,s}=e^{iM^{1/2}Q(s)}$, where $Q$ is an essentially
self-adjoint operator, cf. \cite{berezin}.
Therefore, the method of stationary phase, cf. \cite{maslov} and Section \ref{stationary_phase}, can be applied to extend \eqref{s_part} to the case when
$\lambda_n-V_0$ has finitely many non degenerate critical points $\sigma_k$, 
where the eigenvalue crossings satisfy
\begin{equation}\label{nogap_c}
|\frac{d}{ds}\big(\lambda_n(s)-\lambda_0(s)\big)|_{s=\sigma_k}>c \quad \mbox{
for a positive constant $c$ and $n\ne 0$.}
\end{equation}
With this assumption, the method of stationary phase, see Section \ref{stationary_phase}, implies 
\[
\int_0^t \hat S_{t,s} R\psi_0^\bot(s) ds=
%(2\pi)^{1/2} M^{-1/4}
\mathcal O(M^{-1/4}),
\]
and
\begin{equation}\label{psi_cr}
\|\psi^\bot\|=\mathcal O(M^{-1/4}),
\end{equation}
since by letting the hitting surface $I$ include the points where the electron eigenvalue cross we have
\[
\psi^\bot(t) = \hat S_{t,0} \big( \psi^\perp(0) + M^{-1/4} r \big)
\]
for a smooth and bounded function $r$ of $X$ and for times $t$ close to the time of a path hitting crossing eigenvalues.
To obtain bounds on the derivatives, we note that if the initial data on the hitting surface 
is chosen as $\psi^\perp(0)=-M^{-1/4} r+ \mathcal{O}(M^{-5/4})$, we obtain
\begin{equation}\label{stat_bot}
\|\partial^2\psi^\bot\|_{L^2}=\mathcal{O}(M^{-1/4}),
\end{equation}
since two derivatives of the solution operator increases the result by at most a factor $\mathcal{O}(M)$.

\subsubsection{Difference of  Hamiltonians}
To estimate the difference in the Hamiltonians we use the bounds \eqref{bot_bound} and \eqref{stat_bot} on derivatives
\[\hat\psi\cdot \hat G\sum_k\Delta_{X_k}(\hat G^{-1}\hat\psi) =\mathcal{O}(1),\]
in the Lagrange coordinates. 
%The derivative $\partial\psi^\bot$
%in the Lagrange coordinate satisfies by definition
%$\psi_0\cdot \psi^\bot=0$ and consequently $\psi_0\cdot \partial^2\psi^\bot=0$,
%which implies that 
%\begin{equation}\label{psi_G_term}
%\hat \psi \cdot \partial^2\hat\psi= \psi_0\cdot\partial^2\psi_0 
%+\psi^\bot\cdot\partial^2\psi^\bot
%=\left\{\begin{array}{cc}
%\mathcal O(M M^{-1}) & \mbox{ for a spectral gap}\\
%\mathcal O(M M^{-1}) & \mbox{ for eigenvalue crossings}\\
%\end{array}\right.
%\end{equation}
%since $\partial^k \psi_0=\mathcal O(1)$
%and two derivatives of $\psi^\bot$ 
%yields at most a factor $M$. 
We conclude that
\begin{equation}\label{H_fel}
\begin{array}{ll}
(H_S-H_E)(\hat X,\hat p,\hat\psi)
&=\frac{1}{2M} \hat\psi\cdot \hat G\sum_n\Delta_{X_n}(\hat G^{-1}\hat \psi)=
\left\{\begin{array}{ll}
\mathcal O(M^{-1}) & \mbox{spectral gap}\\
\mathcal O(M^{-1}) & \mbox{crossings}\\
\end{array}\right.\\
(H_E-H_{BO})(\hat X,\hat p,\hat\psi)  
&= \hat\psi\cdot (V-\lambda_0)\hat\psi =\psi^\bot\cdot (V-\lambda_0)\psi^\bot=
\left\{\begin{array}{ll}
\mathcal O( M^{-1}) & \mbox{spectral gap}\\
\mathcal O(M^{-1/2}) & \mbox{crossings.}\\
\end{array}\right.
\end{array}
\end{equation}
The largest bound $\mathcal{O}(M^{-1/2})$ for crossing eigenvalues
can possibly  be improved, since when the eigenvalues cross
we did not use the smallness of $\psi^\bot\cdot (V-V_0)\psi^\bot$  in the neighborhood of
crossing eigenvalues but only
the size of $\psi^\bot=\mathcal{O}(M^{-1/4})$: in a neighborhood of
the crossing point the function $\psi^\perp$ increases from $\mathcal{O}(M^{-1/2})$ outside
to $\mathcal{O}(M^{-1/2})$ inside, for  the component corresponding to the crossing eigenvalue,
while $(V-\lambda_0)\psi^\perp$ vanishes asymptotically at the crossing; a careful
balance of the size of this neighborhood and  of $(V-\lambda_0)^{-1}R\psi_0^\perp$ outside
could give a better estimate of $\psi^\bot\cdot (V-\lambda_0)\psi^\bot$.

\subsubsection{Estimates of the Hamiltonian for the density} \label{imag_rem}
We have $\psi=G^{-1}(\psi_0+\psi^\perp)$. To conclude that
\eqref{imag_alfa} holds, we split into three terms where the first is
\[
\Im \big(G^{-1}\psi_0\cdot \sum_n\Delta_{X_n} (G^{-1}\psi_0)\big) =0
\]
since all functions here are real valued; then we have
\begin{equation}\label{psi_3_4}
M^{-1/2}\Im\big( G^{-1}\psi_0\cdot \sum_n\Delta_{X_n} (G^{-1}\psi^\perp) 
+ G^{-1}\psi^\perp \cdot \sum_n\Delta_{X_n} (G^{-1}\psi_0 \big)=
\left\{\begin{array}{ll}
\mathcal O( M^{-1})  &\mbox{spectral gap}\\
\mathcal O(M^{-3/4}) & \mbox{crossings}\\
\end{array}\right.
%\end{array}
\end{equation}
and
 \[
 M^{-1/2} \Im \big({G^{-1}\psi^\perp} \cdot  \sum_n\Delta_{X_n} (G^{-1}\psi^\perp)\big)= 
 \left\{\begin{array}{ll}
\mathcal O( M^{-3/2})  &\mbox{spectral gap}\\
\mathcal O(M^{-1}) & \mbox{crossings.}\\
\end{array}\right.
 \]
by the derivative bounds \eqref{bot_bound} and \eqref{stat_bot}.
We conclude that
\begin{equation}\label{rho_3_4}
M^{-1/2} \Im (\psi\cdot\sum_n\Delta_{X_n}\psi)/(\psi\cdot\psi)=
\left\{\begin{array}{ll}
\mathcal O( M^{-1})  &\mbox{spectral gap}\\
\mathcal O(M^{-3/4}) & \mbox{crossings.}\\
\end{array}\right.
\end{equation}
 
%\[
%M^{-1/2} \Im G^{-1}\psi^\perp \cdot \sum_n\Delta_{X_n} (G^{-1}\psi_0) =\mathcal{O}(M^{-1})
%\]
%by the bound \eqref{zeta_uppskatt}.
%The remaining term $M^{-1/2} \Im G^{-1}\psi^\perp \cdot \sum_n\Delta_{X_n} (G^{-1}\psi^\perp)$
%requires an additional estimate.
%We have by \eqref{r_int} and \eqref{s_part}
%\[
%\begin{split}
%\psi^\perp(t) &= \hat S_{t,0}\big(\psi^\perp(0) -iM^{-1/2}(V-V_0)^{-1}R\psi_0^\perp(0)\big)
%+ iM^{-1/2}(V-V_0)^{-1}R\psi_0^\perp(t)\\
%& \quad 
%- \underbrace{iM^{-1/2}\int_0^t  \hat S_{t,s}\frac{d}{ds}\big((V-V_0)^{-1}  R\psi_0^\bot\big)(s) ds}_{=\mathcal{O}(M^{-1})},
%\end{split}
%\]
%where the last integral can be integrated by parts again as in \eqref{s_part} to reduce the integral with a factor
%$M^{-1/2}$.
%By choosing $\psi^\perp(0)=-iM^{-1/2}(V-V_0)^{-1}R\psi_0^\perp(0) +\mathcal{O}(M^{-1})$ we have
%\begin{equation}\label{psi_perp_ekv1}
%\psi^\perp(t)=-iM^{-1/2}(V-V_0)^{-1}R\psi_0^\perp(t) + 
%\underbrace{iM^{-1/2}\int_0^t  \hat S_{t,s}\frac{d}{ds}\big((V-V_0)^{-1}  R\psi_0^\bot\big)(s) ds}_{\mathcal{O}(M^{-1})} %\mathcal{O}(M^{-1}),
%\end{equation}
%which implies
%\[
% \sum_n\Delta_{X_n} (G^{-1}\psi^\perp)=-iM^{-1/2}\Delta_{X_n} \big(G^{-1}(V-V_0)^{-1}R\psi_0^\perp\big) +\mathcal{O}(1)= \mathcal{O}(1),
% \]
% since differentiation with respect to $X$ of $\hat S$ gives at most a factor $M^{1/2}$,
% and consequently
% \[
% M^{-1/2} \Im \underbrace{G^{-1}\psi^\perp}_{\mathcal{O}(M^{-1/2})} \cdot \underbrace{
% \sum_n\Delta_{X_n} (G^{-1}\psi^\perp)}_{\mathcal{O}(1)}= \mathcal{O}(M^{-1}).
% \]
%\end{remark}

 \subsection{The stationary phase method}\label{stationary_phase}
 \subsubsection{Real valued phases}
 The theory of semi-classical approximations \cite{maslov,duistermaat} 
 use the method of stationary
 phase  to reduce the study of oscillatory 
 integrals 
 \begin{equation}\label{osc_int}
 \int_{\rset} e^{iM^{1/2}Q(s)} f(s) ds
 \end{equation}
 to points $\sigma$
 where the real valued phase function $Q$ is stationary, i.e. $Q'(\sigma)=0$, as follows;
 such integrals over domains away 
 from the critical points $\sigma_k$, yields by the integration by parts \eqref{s_part} 
  contributions of size $\mathcal O(M^{-1/2})$.
 In a neighborhood of a critical point $\sigma=0$,
 the phase can by Taylors formula by written
  \[
 Q(s)=Q(0) + \frac{1}{2}\underbrace{\int_0^1Q''(ts)dt}_{=:\bar Q(s)}\ s^2.
 \]
Introduce the change of variables 
$\tilde s:= s(\bar Q(s)/\bar Q(0))^{1/2}$ and its inverse map $s(\tilde s)$ to write
\[
\int_{0}^\infty e^{iM^{1/2}s^2\bar Q(s)/2 } f( s) ds
= \int_{0}^\infty e^{iM^{1/2}\tilde s^2\bar Q(0)/2 }  \underbrace{f\big(s(\tilde s)\big) s'(\tilde s)}_{=:\tilde f(\tilde s)}d\tilde s
\]
%Then the error term, which satisfies
% \[
% \begin{array}{ll}
%&  \int_{0}^\infty (e^{iM^{1/2} s^2 \bar Q(s)/2} - e^{iM^{1/2} s^2 \bar Q(0)/2})f(s) ds\\
%&= \int_{0}^\infty e^{iM^{1/2} \tilde s^2 \bar Q(0)/2} \Big(f\big(s(\tilde s)\big) s'(\tilde s) -f(\tilde s)\Big) d\tilde s \\
%&= \int_{0}^\infty e^{iM^{1/2} t \bar Q(0)/2} \Big(f\big(s(t^{1/2})\big) s'(t^{1/2}) -f(t^{1/2})\Big) t^{-1/2} dt \\
%&= 2iM^{-1/2}\bar Q(0)^{-1}\int_{0}^\infty \frac{d}{dt}\big(e^{iM^{1/2} t \bar Q(0)/2} \big)
%\underbrace{\Big(f\big(s(t^{1/2})\big) s'(t^{1/2}) -f(t^{1/2})\Big) t^{-1/2}}_{=:g(t)} dt \\
%&=  -2iM^{-1/2}\bar Q(0)^{-1} g(0)
%-2iM^{-1/2} \bar Q(0)^{-1}\int_{0}^\infty e^{iM^{1/2} t \bar Q(0)/2} g'(t) dt\\
%&  =\mathcal O(M^{-1/2})
%\end{array}
% \]
%% \[
%% \begin{array}{ll}
%%&  \int_{\rset} (e^{iM^{1/2} Q(s)/2} - e^{iM^{1/2} Q(0)/2})f(s) ds\\
%%  &= M^{-1/4}\int_{\rset} (e^{i\tilde s^2 Q(M^{-1/4}\tilde s)/2} 
%%  - e^{i\tilde s^2 Q(0)/2})f(M^{-1/4}\tilde s) d\tilde s\\
%%&  =\mathcal O(M^{-1/2})
%%\end{array}
%%  \]
%since $\|g'\|_{L^1}$ is bounded, becomes lower order   compared to
%the following leading order term;
The corresponding error term based on the integral over the domain $(-\infty,0)$ defines $\tilde f$ on $(-\infty,0)$.
Fourier transformation $\mathcal F$, following \cite{duistermaat}, 
 shows  the expansion
 \begin{equation}\label{stat_fas_metod}
 \begin{array}{ll}
&\int_{\rset} e^{iM^{1/2}s^2Q''(0) } \tilde f( s) ds\\
& = \int_\rset \mathcal F\big(e^{iM^{1/2}s^2Q''(0) }\big)(\omega)
\mathcal F^{-1}\big(\tilde f( s)\big)(\omega) d\omega
\\& =
\int_\rset (2\pi)^{1/2} M^{-1/4} | Q''(0)|^{-1/2} e^{i\pi {\tiny\mbox{sgn}} Q''(0)/4}
e^{-i(Q'')^{-1}M^{-1/2}\omega^2/2}  
\mathcal F^{-1}\big(\tilde f( s)\big)(\omega) d\omega\\
& \sim (2\pi)^{1/2} M^{-1/4} |\, Q''(0)|^{-1/2} e^{i\pi {\tiny\mbox{sgn}} Q''(0)/4}
\sum_{k=0}^\infty \frac{ M^{-k/2}}{k!} (iQ''(0)^{-1}\partial_{ s}^2)^k\tilde f(s)\Big|_{s=0}.
\end{array}
\end{equation}
The original integral \eqref{osc_int} can by a smooth partition of unity be
split into a sum of the integrals over domains containing only one or no critical point.

The stationary phase method can be extended to oscillatory integrals
in $\rset^{3N}$, with the factor $M^{-1/4}$ replaced by $M^{-3N/4}$
and $|Q''(0)|$ replaced by $|\det Q''(0)|$ in \eqref{stat_fas_metod}, cf. \cite{duistermaat,maslov}.

\subsubsection{Self-adjoint phases}
In this section we extend the stationary phase method
to the phase generated by the solution operator $\hat S$.
Since $\hat S$ is unitary, we can write the 
solution operator $\hat S_{t,s}=e^{iM^{1/2}Q_{t,s}}$,
where $Q_{t,s}$ is essentially self-adjoint in $L^2(dx)$.
Therefore $Q_{t,s}$ has real eigenvalues $\{\check\lambda_n(s)\}$
with corresponding orthonormal  real eigenfunctions $\{\check p_n(s)\}$
(satisfying $Q_s\check p_n(s)=\check\lambda_n(s)\check p_n(s)$) that transforms
the operator valued phase to a sum of real valued phases
\[
\int_0^t \hat S_{t,s} R_s ds= \int_0^t e^{iM^{1/2}Q_s} R_s ds
=\sum_n\int_0^t e^{iM^{1/2}\check\lambda_n(s)} 
\underbrace{(\check p_n(s)\cdot R_s)
\ \check p_n(s)}_{=(\check p_n(t)\cdot R_s)\ \check p_n(t)} ds
\]
and the stationary phase method can now be applied to each term separately.

Note that $\check p_n(t)$ does not depend on $s$, since the data
for $\check p_n$ is that they are the eigenfunctions of $\dot Q_{t,t}=V-V_0$ at time $t$:
the eigenvalue relation
\[
(t-s)^{-1}Q_{t,s}\check p_n(s)=(t-s)^{-1}\check\lambda_n(s)\check p_n(s)
\]
and $(t-s)^{-1}Q_{t,s}\rightarrow V-V_0$ as $s\rightarrow t$
shows that $\check p_n(t)$ are the eigenfunctions to $V-V_0$;
let $f_n(t):=e^{iM^{1/2}\check\lambda_n(s)}\check p_n(s)$, then 
the continuity at $t$ implies $f_n(t)=\check p_n(t)$, so that 
$\check p_n(t)=e^{iM^{1/2}\check\lambda_n(s;t)}\check p_n(s;t)$, which
yields $(\check p_n(s)\cdot R_s)
\ \check p_n(s)=(\check p_n(t)\cdot R_s)\ \check p_n(t)$.

The definition of $Q_s$ implies that
\[
\begin{array}{ll}
\frac{d}{ds}Q_s &=V(\hat X_s)-V_0(\hat X_s)\\
\frac{d^2}{ds^2}Q_s &=\frac{d}{ds}\big( V(\hat X_s)-V_0(\hat X_s)\big),
\end{array}
\]
and by differentiation of the  eigenvalue relation we obtain
\[
\frac{d}{ds}\check\lambda_n(s)=\check p_n\cdot (V-V_0)\check p_n(s).
\]
When $t\rightarrow s$ we obtain $\dot{\check\lambda}_n=\lambda_n$
and $\ddot{\check\lambda}_n=\dot\lambda_n$, so that
the condition \eqref{gap_cr1} for crossing eigenvalues of $\lambda_n$
yields the condition for non degenerate critical points of $\check\lambda_n$.

Our analysis shows that 
\[
\check p_0(t)=e^{iM^{1/2}\check\lambda_n(s;t)}\check p_0(s)=\hat S_{t,s}\check p_0(s)\]
which by \eqref{psi_cr} satisfies 
$\check p_0=\psi_0 +\mathcal O(M^{-1/4})$.
Analogously, with a similar assumption on non degenerate
critical points in $\check\lambda_m-\check\lambda_n$, for all $m$ and $n$,
 we obtain for the other eigenfunctions
$e^{iM^{1/2}\int_s^t V_n(r)-V_0(r) dr}\check p_n(s)= \psi_n +\mathcal O(M^{-1/4})$,
and we can conclude that $\check p_n$ is $\mathcal O(M^{-1/4})$ close to the electron
eigenfunction $\psi_n$.

 \subsubsection{Analysis of the Schr\"odinger dynamics.} 
 The analysis above can be  applied to the exact solution operator 
 $S_{t,s}=e^{iM^{1/2}Q_s}$ and $\tilde V$
 replacing the Ehrenfest operator $\hat S$ and $V-V_0$.
 Let $\tilde\psi=\psi_0\oplus\tilde\psi^\bot$, then
$\tilde\psi_t=S_{t,0}\tilde\psi_0$
and $R\psi_0=\dot \psi_0 - M^{1/2}\tilde V\psi_0$. 
Integrate by parts to obtain
\begin{equation}\label{S_osc}
\begin{array}{ll}
\int_0^t S_{t,s} R\psi_0^\bot(s) ds &= 
\int_0^t \underbrace{\frac{d}{ds} (S_{t,s})iM^{-1/2}\tilde V^{-1}}_{=S_{t,s}}  R\psi_0^\bot(s) ds \\
&=iM^{-1/2}\Big[S_{t,s}\tilde V^{-1}  R\psi_0^\bot(s)\Big]_{s=0}^t \\
&\quad -
iM^{-1/2}\int_0^t  \hat S_{t,s}\frac{d}{ds}\big(\tilde V^{-1}  R\psi_0^\bot\big)(s) ds
\end{array}
\end{equation}
and use
\[
\begin{array}{ll}
\tilde V^{-1}R\psi_0^\bot
&=\big(I +(\tilde V-V-V_0)(V-V_0)^{-1}\big)^{-1}(V-V_0)^{-1}R\psi_0^\bot\\
&=\sum_{n\ne 0} \big(I +(\tilde V-V-V_0)(V-V_0)^{-1}\big)^{-1} (\lambda_n -V_0)^{-1}
\psi_n R\psi_0^\bot\cdot \psi_n\\
&=\sum_{n\ne 0}  (\lambda_n -V_0)^{-1} \psi_n R\psi_0^\bot\cdot \psi_n 
+\mathcal O(M^{-1})
\end{array}
\]
to deduce as before that $\tilde\psi^\bot=\mathcal O(M^{-1/2})$
when the spectral gap condition holds and $\psi_n$ are smooth.

In the case of crossing eigenvalues, the method of stationary phase,
projected to the eigenfunction of the singularity,
 is applicable with  $S$, in the weighted scalar product
 $\int v\cdot w G^{-2} dX$ where the symmetric generator is $(V-V_0)G^{-1}\tilde\psi
 -(2M)^{-1} \sum_n\Delta_{X_n} (G^{-1}\tilde\psi)$.
 In the new variable $\psi=G^{-1}\tilde\psi$ the operator
 becomes $V-V_0  -(2M)^{-1} \sum_n\Delta_{X_n}$
 with its time derivative equal to $d(V - V_0)/dt$, as in the Ehrenfest case.
 The method of stationary phase for $S$ 
 can therefore be used to study the behavior of the oscillatory integral
 $\int_0^t S_{t,s}R\psi_0^\bot$ in the subspace of a
 single eigenfunction $\psi_n$,  in a neighborhood 
 a critical point where $\lambda_n-V_0$ vanish, $n\ne 0$;
 in the orthogonal subspace to $\psi_n$ and $\psi_0$,
 the oscillatory integral can be estimated by integration by
 parts as in \eqref{S_osc}.
 With such a projection into
 a single eigenfunction direction, the source term is 
 $R\psi_0^\bot\cdot \psi_n$ and  the non degeneracy condition for  the stationary phase method becomes 
 \[
 |\psi_n\cdot \ddot Q(0)\psi_n|=|\psi_n\cdot d(V-V_0)/dt\psi_n|
 =|\dot\lambda_n -\dot V_0|>c.
 \]
 Note that the projection to a one dimensional subspace is
 possible when the critical points of the different eigenvalues are
 separated. If two or more eigenvalues have the same critical point,
 a more careful study would be needed.

 In conclusion we have the same estimates
 for $\tilde\psi=:\psi_0\oplus \tilde\psi^\bot $ as for $\hat\psi$:
 \begin{equation}\label{s_bound}
 \begin{array}{ll}
& \|\tilde\psi^\bot\|_{L^2}=\mathcal O(M^{-1/2}) \\
\end{array}
\end{equation}
with a spectral gap,  respectively 
 \begin{equation}\label{s_bound_cr}
 \begin{array}{ll}
& \|\tilde\psi^\bot\|_{L^2}=\mathcal O(M^{-1/4}) \\ 
\end{array}
\end{equation}
with non degenerate eigenvalue crossings, and
as in \eqref{psi_3_4}, we also obtain 
\begin{equation}\label{S_g_termer}
\frac{1}{2M} \tilde\psi\cdot \sum_n\Delta_{X_n}(G^{-1}\tilde\psi) =
\left\{ 
\begin{array}{ll}
\mathcal O(M^{-1}) & \mbox{ with a spectral gap}\\
\mathcal O(M^{-1}) & \mbox{ with eigenvalue crossings.}
\end{array}\right.
\end{equation}

\subsection{Error estimates for the densities}
The densities are
\[\begin{array}{lll}
\rho_S &= G^{-2}\tilde\psi\cdot\tilde\psi &\mbox{for the Schr\"odinger equation}\\
\rho_E &= \hat G^{-2} &\mbox{for the Ehrenfest dynamics}\\
\rho_{BO} &= \bar G^{-2} &\mbox{for the Born-Oppenheimer dynamics.}
\end{array}\]
We have from \eqref{H_fel}, (\ref{s_bound}-\ref{S_g_termer}) and \eqref{g_fel}
\[\begin{array}{ll}
G^{-2}&= \hat G^{-2} +\mathcal O(M^{-\alpha})\\
G^{-2}&= \bar G^{-2} +\mathcal O(M^{-\alpha})
\end{array}\]
and by (\ref{s_bound}-\ref{s_bound_cr})
\begin{equation}\label{tilde_psi_m}
\tilde\psi\cdot\tilde\psi= 1 + \mathcal O(M^{-\alpha}),
\end{equation}
which proves
\[
\begin{array}{ll}
\rho_S &=\rho_E +\mathcal O(M^{-\alpha})\\
\rho_S &=\rho_{BO} +\mathcal O(M^{-\alpha}),
\end{array}\]
where 
\[
\alpha = \left\{ \begin{array}{ll}
1 & \mbox{for the spectral gap condition \eqref{gap_c1}}\\
1/2 & \mbox{for non degenerate eigenvalue crossings \eqref{gap_cr1}.} \\
\end{array}
\right.
\]
The sharper estimate for the density obtained using  \eqref{g_fel} and \eqref{rho_3_4}, 
instead of \eqref{tilde_psi_m}, improves the error bound to
$\mathcal{O}(M^{-3/4})$ in \eqref{tilde_psi_m}
for the case of Ehrenfest dynamics approximating Schr\"odinger observables with crossing electron eigenvalues.

\section{Approximation of  stochastic WKB states by stochastic dynamics}\label{md_stok}
In this section we analyze a situation when the WKB eigenstate,
(corresponding to a degenerate eigenvalue of the Schr\"odinger equation
\eqref{schrodinger_stat})
is perturbed from the corresponding electron ground state by thermal fluctuations, which will lead to
molecular dynamics in the canonical ensemble with constant number of particles, volume and temperature.
To determine the stochastic data for the Schr\"odinger wave function $\phi$ in \eqref{s-hamilton}
requires some additional assumptions.
Inspired by the study of a classical heat bath
of harmonic oscillators in \cite{zwanzig},
we will sample $\phi$ randomly from a probability density given by an {\it equilibrium solution}
$f$ (satisfying $\partial_t f=0$) of the {\it Liouville equation} $\partial_t f +\partial_{p_S}H_S \partial_{r_S}f -
\partial_{r_S}H_S \partial_{p_S}f =0$,  to the Schr\"odinger Hamiltonian dynamics \eqref{s-hamilton}. 
There are many such equilibrium solutions, e.g. $f= h( H_S)$
for any differentiable function $h$ and there may also be equilibrium densities that are not functions of the Hamiltonian.
The Hamiltonians for the Schr\"odinger and Ehrenfest dynamics differ by the term 
\[
\frac{1}{M} \Re\big( \phi\cdot G\sum_n\Delta_{X_n}(G^{-1}\phi)\big)
\]
which is small of $\phi$ is smooth enough, as was verified in Theorems \ref{thm_ehrenfestapprox} and \ref{bo_thm}.
 Since the Ehrenfest Hamiltonian is simpler to understand,
in particular showing that smooth $\phi$ are most probable, we  first focus on the Ehrenfest dynamics.

To find the  equilibrium solution to sample $\phi$  from, we may first consider
the marginal equilibrium distribution for the nuclei in the Ehrenfest Hamiltonian system.
The equilibrium distribution of the nuclei is simpler to understand than the equilibrium of electrons: 
in a statistical mechanics view, the conditinal probability
of finding the now classical nuclei with the energy $H_E:=2^{-1}|p|^2 + \phi\cdot V(X)\phi$, for given $\phi$,
in \eqref{E-hamiltonian}
is proportional to
the { Gibbs-Boltzmann factor}  $\exp(-H_E/T)$, where the positive parameter $T$ is the
temperature, in units of the Boltzmann constant,
cf. \cite{feynman}, \cite{kadanoff}. Assuming that the equilibrium density is
a function of the Hamiltonian and that the
marginal distribution for $p$  is the Boltzmann distribution proportional to $e^{-|p|^2/(2T)}$, 
its marginal distribution  therefore satisfies
\[
\int h\big(2^{-1}|p|^2 + \phi\cdot V(X)\phi\big) d\phi_rd\phi_i = e^{-|p|^2/(2T)}C(X)
\] for some function $C:\rset^{3N}\rightarrow \rset$ and Fourier transformation with
respect to $|p|\in\rset$ of this equation implies the {\it Gibbs distribution}
\[
h(H_E)= ce^{-H_E/T},
\]
for a normalizing constant $c=1/\int \exp(-H_E/T) d\phi_rd\phi_idXdp$.
In this perspective
{\it the nuclei act as  the heat  bath for the electrons}. 
The work \cite{kozlov} considers Hamiltonian systems  where the equilibrium densities are assumed to be
a function of the Hamiltonian and shows that the first and second law of (reversible) thermodynamics
hold (for all Hamiltonians depending on a multidimensional set of parameters)
if and only if the density is the Gibbs exponential $\exp(-H_E/T)$ (the "if" part was formulated already
by Gibbs \cite{gibbs});  in this sense, the
Gibbs distribution is more stable than other equilibrium solutions.
An alternative motivation of the Gibbs distribution, based on the conclusion of this work (in a somewhat circular argument), is
that the nuclei can be approximated by classical Langevin dynamics with the
 unique invariant density $\exp\big(-|p|^2/2 -\lambda_0(X)\big)/T$, which is an accurate
approximation of the marginal distribution of $\exp(- H_E/T)$ when integrating over all the electron states $\phi$,
see Lemma \ref{r_X} and Theorem \ref{thm_md_stok}.
Note that there is only one function of the Hamiltonian where the momenta $p_j$ are independent
and that is when 
$f(r_E,p_E)$ is proportional to $\exp{(- H_E/T)}$. 

Since the energy is conserved
for the Ehrenfest dynamics  --now viewed with the electrons as the primary systems coupled to
the heat bath of nuclei-- the probability of finding the electrons in a certain configuration $\phi$ is
the same as finding the nuclei in a state with energy $ H_E$, 
which is proportional to $\exp(- H_E/T)$ in the canonical ensemble.
 This conclusion, that the probability to have an electron wave function
$\phi$ is proportional to $\exp(- H_E/T) d\phi^r d\phi^i$ 
is our motivation to
sample the  data for $\phi$ from the conditioned density generated be 
$\exp{(-\phi\cdot V(X)\phi/T)}d\phi^r d\phi^i$:
since we seek data for the electrons, we use the probability distribution for $\phi$ conditioned
on $(X,p)$. 

In \eqref{obser_decay} we established the asymptotic density 
$\sum_n \hat\rho_n r_n$ 
for a multiple state system with probability $r_n$ for eigenstate $n$, as defined in \eqref{surf_hopp}.
The case of the Gibbs ensemble therefore yields 
the asymptotic relation
\begin{equation}\label{rho_r}
\frac{\sum_n \hat\rho_n r_n}{\sum_n r_n} \simeq \frac{\int e^{-H_E/T} dp d\phi}{\int e^{-H_E/T} dp d\phi dX}.
\end{equation}

We compare in Remark \ref{entropy} our model of stochastic data
with a more standard model having given probabilities to be in mixed states,
which is not an equilibrium solution of the Ehrenfest dynamics.
To sample from the Gibbs equilibrium density is standard in classical Hamiltonian
statistical mechanics but it seems non standard for Ehrenfest quantum dynamics.

\subsubsection{The Constrained Stochastic  Data}\label{sec:initial_data}
As in models of heat baths \cite{ford_kac1,ford_kac2} and \cite{zwanzig} we assume that the  data of
the light particles (here the electrons)
are stochastic, sampled from an equilibrium distribution of the Liouville equation.
All states in this distribution correspond to pure eigenstates
of the full Schr\"odinger operator with energy $E$.
There are many such states
and here we use the canonical ensemble where the data is in state $\phi$
with the Gibbs-Boltzmann distribution proportional to
$\exp{(- H_E/T)}d\phi^rd\phi^i dp\, dX$,
i.e. in any state $\phi$, for $\|\phi\|=1$,  with probability weight
\[
{e^{- \phi\cdot V\phi/T}d\phi^rd\phi^i}.
\]
Let us now determine precise properties of this distribution generated by
the Hamiltonian $H_E$. 
To reduce the complication of the constraint $\phi\cdot\phi=1$, we change variables
$\phi=\tilde\phi/(\tilde\phi\cdot\tilde\phi)^{1/2}$ and write the Hamiltonian equilibrium density as
\begin{equation}\label{gibbs_1}
\exp\Big(- \big(p\bullet p/2 + \lambda_0 + \frac{\tilde\phi\cdot (V-\lambda_0)\tilde\phi}{\tilde\phi\cdot\tilde\phi}\big) /T\Big) d\tilde\phi^r d\tilde\phi^i dp \, dX.
\end{equation}
Diagonalize the electron operator $V(X^t)$
by the
normalized eigenvectors and eigenvalues $\{\bar p_j,\lambda_j\}$
\[
 \tilde\phi \cdot V(X^t)\tilde\phi  = \lambda_0 + \sum_{j>0} \underbrace{( \lambda_j -\lambda_0)}_{=:\, \bar\lambda_j} |\gamma_j|^2
\]
where 
\begin{equation}\label{summa_p}
\begin{array}{ll}
\tilde\phi&=\sum_{j\ge 0} \gamma_j \bar p_j ,\\
(V-\lambda_0)\bar p_j&= \bar \lambda_j \bar p_j,\\
\bar\lambda_0&=0,
\end{array}
\end{equation}
with real and imaginary parts $\gamma_j=:\gamma_j^r + i\gamma_j^i$.
The orthogonal transformation $\tilde\phi=\sum_j\gamma_j \bar p_j$ shows that the probability density \eqref{gibbs_1}
is given by 
\begin{equation}\label{D}
\mathcal D:=\frac{\Big( \prod_{j\ge 0} e^{-\bar\lambda_j|\gamma_j|^2/(T\sum_{j\ge 0}|\gamma_j|^2)}
 d\gamma_j^r \, d\gamma_j^i\Big) 
\ e^{-(p\bullet p/2 +\lambda_0(X))/T}dp\, dX}
{\int_{\rset^{6N}}
\Big(\prod_{j\ge 0} \int_{\rset^2}e^{-\bar\lambda_j|\gamma_j|^2/(T\sum_{j\ge 0}|\gamma_j|^2)} d\gamma_j^r \, d\gamma_j^i \Big)
\ e^{-(p\bullet p/2 +\lambda_0(X))/T}dp\, dX}\ ,
\end{equation}
using that the determinant of
the matrix of eigenvectors is one. 

If we neglect the constraint and set $\sum_{j\ge 0}|\gamma_j|^2=1$ the joint distribution density $\mathcal D$ is approximated by
\[
 \frac{\Big(\prod_{j> 0} e^{-\bar\lambda_j|\gamma_j|^2/T} d\gamma_j^r \, d\gamma_j^i \Big)
\ e^{-(p\bullet p/2 +\lambda_0(X))/T}dp\, dX}
{\int_{\rset^{6N}}
\Big(\prod_{j> 0} \int_{\rset^2}e^{-\bar\lambda_j|\gamma_j|^2/T} d\gamma_j^r \, d\gamma_j^i \Big)
\ e^{-(p\bullet p/2 +\lambda_0(X))/T}dp\, dX},
\]
where  $\{\gamma_j^r,\gamma_j^i, j>0\}$ are  independent and each $\gamma_j^r$ and $\gamma_j^i$ is
normally distributed  with mean zero and variance $T/\bar\lambda_j$.
We see in Lemma \ref{r_X} that this approximation is
accurate, provided the spectral gap conditions
\begin{equation}\label{gap_cond}
\begin{array}{ll}
&\frac{T}{\bar\lambda_1} \ll 1,\\
&\min_{X_c\in D}\max_{
\tiny\begin{array}{c}
X\in D \\
Y=sX+(1-s)X_c\\
s\in[0,1]
\end{array}
}\Big|\sum_{j>0}\frac{\partial_X\bar\lambda_j(Y) \bullet (X-X_c)}{\bar\lambda_j(Y)}\Big|=:\kappa\ll 1 % \ll 1
\end{array}
\end{equation}
hold, where $D$ is the set of attained nuclei positions for the Ehrenfest dynamics.
The first condition implies that the temperature is small compared to the gap.
The second condition means that the electron eigenvalues are almost parallel 
to the ground state eigenvalue, as a function of $X$.
We have $\partial_X\bar\lambda_j=\bar p_j\cdot \partial_X (V-\lambda_0)\bar p_j$
which typically decays fast with growing eigenvalues, since $\partial_X (V-\lambda_0)$ is smooth
and $\bar p_j$ become highly oscillatory as $j\rightarrow\infty$.

\begin{lemma}\label{r_X} Assume the electron eigenvalues have a
 spectral gap around the ground state eigenvalue $\lambda_0$,
 in the sense that \eqref{gap_cond} holds. Then the marginal probability mass
\[
r(X):=\prod_{j\ge 0} \int_{|\gamma_j|^2<C}e^{-\bar\lambda_j|\gamma_j|^2/(T\sum_{j\ge 0}|\gamma_j|^2)} d\gamma_j^r \, d\gamma_j^i  
\]
satisfies
\begin{equation}\label{r_log_bound}
|\log \frac{r(X)}{r(X_c)}| \le  \kappa\, .
\end{equation}
\end{lemma}

\begin{proof} We first note that $\|\phi\|_{L^2}$ is bounded so that each component $|\gamma_j|$ is also
bounded.
Each integral factor has the derivative
\[
\begin{array}{ll}
&\partial_X\int_{|\gamma_n|^2<C}e^{-\bar\lambda_n|\gamma_n|^2/(T\sum_{j\ge 0}|\gamma_j|^2)} d\gamma_n^r \, d\gamma_n^i \\
&=\frac{\partial_X\bar\lambda_n(X)}{\bar\lambda_n(X)}\int_{|\gamma_n|^2<C} \frac{
\bar\lambda_n|\gamma_n|^2}{T\sum_{j\ge 0}|\gamma_j|^2}
e^{-\bar\lambda_n|\gamma_n|^2/(T\sum_{j\ge 0}|\gamma_j|^2)} d\gamma_n^r \, d\gamma_n^i \\ 
\end{array}
\]
for $n\ne 0$ and the derivative equal to zero for $n=0$, so that 
\[
\begin{array}{ll}
&\partial_Xr(X) \\
&= r(X)\sum_{n> 0}  \frac{\partial_X\bar\lambda_n(X)}{\bar\lambda_n(X)}\ 
\frac{\int_{|\gamma_n|^2<C} \frac{
\bar\lambda_n|\gamma_n|^2}{T\sum_{j\ge 0}|\gamma_j|^2}
e^{-\bar\lambda_n|\gamma_n|^2/(T\sum_{j\ge 0}|\gamma_j|^2)} d\gamma_n^r \, d\gamma_n^i }
{\int_{|\gamma_n|^2<C}  
e^{-\bar\lambda_n|\gamma_n|^2/(T\sum_{j\ge 0}|\gamma_j|^2)} d\gamma_n^r \, d\gamma_n^i }.
\end{array}
\]
The integral in the denominator has the estimate
\[
\begin{array}{ll}
&\int_{|\gamma_n|^2<C}e^{-\bar\lambda_n|\gamma_n|^2/(T\sum_{j\ge 0}|\gamma_j|^2)} d\gamma_n^r \, d\gamma_n^i \\ 
&=  \frac{T}{\bar\lambda_n}  \sum_{j\ne n}|\gamma_j|^2
\int_{|\hat\gamma_n|^2<C'\bar\lambda_n/T}e^{-|\hat\gamma_n|^2/(1 + \bar\lambda_n^{-1} T|\hat\gamma_n|^2)} d\hat\gamma_n^r \, d\hat\gamma_n^i  \\
&=  \frac{T}{\bar\lambda_n}\sum_{j\ne n}|\gamma_j|^2\Big(
\int_{|\hat\gamma_n|^2<\epsilon\bar\lambda_n/T}e^{-|\hat\gamma_n|^2/(1 + \bar\lambda_n^{-1} T|\hat\gamma_n|^2)} d\hat\gamma_n^r \, d\hat\gamma_n^i\\
&\quad   + \mathcal{O}(\bar\lambda_n T^{-1}e^{-\bar\lambda_n\epsilon/(2T)})\Big)\\
\end{array}
\]
using $\epsilon\bar\lambda_n/T<|\hat\gamma_n|^2<C'\bar\lambda_n/T$ and $T/\bar\lambda_n\ll 1$ in this
domain; the integral over the remaining domain satisfies
\[
\begin{array}{ll}
\int_{|\hat\gamma_n|^2<\epsilon\bar\lambda_n/T}e^{-|\hat\gamma_n|^2} d\hat\gamma_n^r \, d\hat\gamma_n^i  
&\le \int_{|\hat\gamma_n|^2<\epsilon\bar\lambda_n/T}e^{-|\hat\gamma_n|^2/(1 + \bar\lambda_n^{-1} T|\hat\gamma_n|^2)} d\hat\gamma_n^r \, d\hat\gamma_n^i  \\
&\le (1+\epsilon)
\int_{|\hat\gamma_n|^2<\epsilon\bar\lambda_n/(T(1+\epsilon))}e^{-|\hat\gamma_n|^2} d\hat\gamma_n^r \, d\hat\gamma_n^i  .
\end{array}
\]
Choose $\epsilon=4T\bar\lambda_n^{-1}\log(\bar\lambda_n/T)$ to obtain
\begin{equation}\label{int_r1}
\int_{|\gamma_n|^2<C}e^{-\bar\lambda_n|\gamma_n|^2/(T\sum_{j\ge 0}|\gamma_j|^2)} d\gamma_n^r \, d\gamma_n^i \\ 
=  \frac{T}{\bar\lambda_n} \sum_{j\ne n}|\gamma_j|^2 \Big({\pi}+ \mathcal{O}\big(T\bar\lambda_n^{-1}  
\log (\bar\lambda_n T^{-1})\big)\Big).
\end{equation}
We have similarly 
\begin{equation}\label{int_r2}
\begin{array}{ll}
&\int_{|\gamma_n|^2<C}\frac{
\bar\lambda_n|\gamma_n|^2}{T\sum_{j\ge 0}|\gamma_j|^2}
e^{-\bar\lambda_n|\gamma_n|^2/(T\sum_{j\ge 0}|\gamma_j|^2)} d\gamma_n^r \, d\gamma_n^i \\ 
&=  \frac{T}{\bar\lambda_n}\sum_{j\ne n}|\gamma_j|^2 \Big(\frac{\pi}{2} + \mathcal{O}\big(T\bar\lambda_n^{-1}  
\log (\bar\lambda_n T^{-1})\big)\Big),
\end{array}
\end{equation}
which implies that
\begin{equation}\label{int_r3}
|\log \frac{r(X)}{r(X_c)}|
=|\int_0^1\frac{\partial_Xr}{r}\big(sX_c+(1-s)X\big)\bullet (X-X_c) \, ds|\le \kappa.
\end{equation}
\end{proof}

\begin{remark}{\it Entropy and the Standard Canonical Density Distribution.}\label{entropy}
Let $q_j$ denote the probability of electron state $j$ in the  data $\phi$.
In the setting of the canonical  distribution model there holds
\begin{equation}\label{canon}
q_j={e^{-\bar\lambda_j/T}}/{\sum_je^{-\bar\lambda_j/T}},
\end{equation}
which follows from maximizing the von Neumann entropy defined by $-\sum_j q_j\log q_j$, with the probability and energy constraints
$\sum_j q_j=1$ and $\sum_j\bar\lambda_jq_j=\mbox{ constant},$  see \cite{gardiner}.
The stochastic model for the variable $|\gamma_j|^2$,
measuring  in \eqref{D} and \eqref{summa_p} the probability to be in electron state $j$,
is  different from the  model \eqref{canon}, since the Gibbs distribution 
$\frac{\int e^{-H_E/T} dp d\phi}{\int e^{-H_E/T} dp d\phi dX}=\sum_n \hat\rho_n r_n/\sum_n r_n $
in \eqref{summa_p}  includes both the density $\hat\rho_n$ and the weight $r_n:=|\gamma_n|^2$
to be in electron state $n$.

\end{remark}

\subsubsection{The stochastic molecular dynamics models}
In this subsection we consider the special case when the observable does not
depend on the time variable  (and the velocity). Then it is enough to determine
an integral with respect to the invariant measure; there are several alternatives, cf. \cite{cances},
and we focus on ensemble averages computed by stochastic Langevin and Smoluchowski dynamics.
The observable in the Ehrenfest dynamics is then
\begin{equation}\label{gibbs_ehren}
\begin{array}{ll}
\frac{\int_{\rset^{3N}} g(X)e^{-\lambda_0(X)/T} r(X) dX}{\int_{\rset^{3N}} e^{-\lambda_0(X)/T} r(X)dX}
&=\frac{\int_{\rset^{3N}}g(X)e^{-\lambda_0(X)/T}\  \frac{r(X)}{r(X_c)}  dX}{\int_{\rset^{3N}} e^{-\lambda_0(X)/T} \ \frac{r(X)}{r(X_c)} dX}\\
&=\frac{\int_{\rset^{3N}}g(X)e^{-\lambda_0(X)/T +\log \frac{r(X)}{r(X_c)}}  dX}
{\int_{\rset^{3N}} e^{-\lambda_0(X)/T +\log\frac{r(X)}{r(X_c)}} dX},\\
\end{array}
\end{equation}
where by Lemma \ref{r_X}
\[
 |\log\frac{r(X)}{r(X_c)}| \le \kappa,
\]
which implies
\begin{equation}\label{gibbs_obser}
\frac{\int_{\rset^{3N}} g(X)e^{-\lambda_0(X)/T} r(X) dX}{\int_{\rset^{3N}} e^{-\lambda_0(X)/T} r(X) dX}
=\frac{\int_{\rset^{3N}} g(X)e^{-\lambda_0(X)/T} dX}{\int_{\rset^{3N}} e^{-\lambda_0(X)/T}  dX} + \mathcal{O}(\kappa).
\end{equation}

Let $W_t$ denote the standard Brownian process (at time $t$)
in $\rset^{3N}$ with independent components and let $K$ be any positive parameter.
The stochastic  Langevin dynamics
\[
\begin{array}{ll}
dX_t &= p_t dt\\
dp_t &= -\partial_X\lambda_0(X_t) dt - Kp_t dt + \sqrt{2TK}dW_t
\end{array}
\]
and the Smoluchowski dynamics
\[
dX_s= -\partial_X\lambda_0(X_s) ds  + \sqrt{2T}dW_s
\]
has the unique invariant probability density
\[
\frac{e^{-(p\bullet p/2 +\lambda_0(X))/T} dp\, dX}{\int_{\rset^{6N}} e^{-(p\bullet p/2 +\lambda_0(X))/T} dp\, dX}
\]
respectively 
\[
\frac{e^{-\lambda_0(X)/T} dX}{\int_{\rset^{3N}} e^{-\lambda_0(X)/T} dX},
\]
cf. \cite{cances}. We see that the invariant measure for the Smoluchowski dynamics
and the marginal invariant measure of the Langevin dynamics are equal to
the first term in the right hand side of \eqref{gibbs_obser} and we have proved
\begin{theorem}\label{thm_md_stok}
Both Langevin and Smoluchowski stochastic molecular dynamics
approximate Gibbs Ehrenfest observables \eqref{gibbs_ehren} with error bounded by 
$\mathcal{O}(\kappa)$, provided 
the spectral gap condition \eqref{gap_cond} holds. 
\end{theorem}

If we assume that 
\begin{equation}\label{pp_bound}
\mbox{the matrix $\bar p_n\cdot G\sum_k\Delta_{X_k}(G^{-1}\bar p_m)$ is  (uniformly in $X$) 
bounded in $\ell^2$,}
\end{equation}
then the difference of the Schr\"odinger and Ehrenfest Hamiltonians
has the bound
\[
\begin{array}{ll}
M^{-1}|\Re\big( \phi\cdot G\sum_k\Delta_{X_k}(G^{-1}\phi)\big)|
&=M^{-1}|\Re\big( \sum_{n,m}\gamma_n^*\gamma_m \bar p_n\cdot G\sum_k\Delta_{X_k}(G^{-1}\bar p_m)\big)|\\
&\le \mathcal O(M^{-1}) \sum_n |\gamma_n|^2,
\end{array}
\]
which is asymptotically negligible compared to the
Ehrenfest potential energy
\[
\phi\cdot V\phi= \sum_n \bar\lambda_n|\gamma_n|^2,
\]
and we obtain
\begin{corollary}\label{thm_md_s}
Both Langevin and Smoluchowski molecular dynamics
approximate Gibbs Schr\"odinger observables 
\[\int g(X) e^{-H_S(X,p,\phi)/T} dXdpd\phi/ \int e^{-H_S(X,p,\phi)/T} dXdpd\phi\]
 with error bounded by 
$\mathcal{O}(M^{-1} + \kappa)$, 
provided  \eqref{gap_cond} and \eqref{pp_bound} hold.
\end{corollary}

By combining \eqref{int_r1} and \eqref{int_r2},  the proof in Lemma \ref{r_X} in fact shows
\[
\partial_X \log r(X)= \frac{T}{2}\sum_{n>0}\frac{\partial_X\bar\lambda_n}{\bar\lambda_n}
+ \mathcal O(1)\sum_{n>0}|\partial_{X_j}\bar\lambda_n|\bar\lambda_n^{-2} \log \bar\lambda_n^{-1}
\]
which together with  \eqref{int_r3} and \eqref{gibbs_ehren} imply
\begin{theorem}\label{thm3}
If we instead of the gap condition \eqref{gap_cond}
assume the weaker condition
\begin{equation}\label{trace_gap}
\begin{array}{ll}
T\bar\lambda_n^{-1} &\ll 1\\
\sum_{n>0}|\partial_{X_j}\bar\lambda_n|\bar\lambda_n^{-1}  &=\mathcal O(1),\\
\sum_{n>0}|\partial_{X_j}\bar\lambda_n|\bar\lambda_n^{-2} \log \bar\lambda_n^{-1} &=\kappa\ll 1,
\end{array}
\end{equation}
then the results in Theorem \ref{thm_md_stok} and Corollary \ref{thm_md_s} hold, provided the Born-Oppenheimer potential  
$\lambda_0$ is replaced
by
\[
\lambda_0(X) + \frac{T}{2}\sum_{n>0} \log \bar\lambda_n(X).
\]
\end{theorem}
Note that the correction can be written as 
\[
\frac{T}{2}{\rm trace}^\bot\log (V-\lambda_0)
\]
where ${\rm trace}^\bot$ is the trace in the orthogonal complement of the electron ground state $\psi_0$
satisfying $V\psi_0=\lambda_0\psi_0$.

The work \cite{ASz2}  shows that Langevin dynamics, using
the rank one friction and diffusion matrix  \[K=K(X)= 2M^{-1/2}\partial_X\psi_0(X)\cdot\partial_X\psi_0(X),\] 
approximates time-dependent observables based on the Ehrenfest dynamics with the  accuracy $o(M^{-1/2})$ on bounded time intervals
if $\kappa=O(M^{-1/2})$.

Theorem \ref{thm_md_stok} 
is relevant for the 
central problem in statistical mechanics 
to show that  Hamiltonian dynamics of heavy particles, coupled to a heat bath of many lighter particles
with random initial data, can be approximately described by Langevin's equation, as motivated by
the pioneering work \cite{einstein},\cite{langevin} 
and continued with more precise
heat bath models, based on harmonic interactions, in \cite{ford_kac1,ford_kac2} \cite{zwanzig}.

\section{Discrete spectrum}
\label{fredholm} 
This section verifies that the symmetric bilinear form 
\[
\int_{\tset^{3(J+N)-1}} v \bar V_\tau v +\gamma v^2 \ dxdX_0
\]
is continuous and coercive on $H^1(\tset^{3(J+N)-1})$ also for the Coulomb potential, which
implies that the spectrum of $\bar V$ is discrete by
the theory of compact operators, see \cite{evans}.
Let $r:=|x^j-X^n|$. Integrate by parts, for any $\epsilon>0$, to obtain
\[
\begin{array}{ll}
-\int_0^R \frac{1}{|x^j-X^n|}v^2 r^2 dr 
&=-\int_0^R  \frac{1}{r} v^2 r^2 dr \\
&=-\int_0^R   v^2\partial_r \frac{r^2}{2}  dr \\
&=\int_0^R  v \partial_r v \  r^2 dr + [\frac{v^2 r^2}{2}]_{r=0}^{r=R}\\
&\ge -\Big(\int_0^R v^2 r^2 dr \int_0^R (\partial_r v)^2 r^2 dr \Big)^{1/2} 
+ [\frac{v^2 r^2}{2}]_{r=0}^{r=R}\\
&\ge -\frac{1}{2\epsilon} \int_0^R v^2 r^2 dr 
-\frac{\epsilon}{2} \int_0^R (\partial_r v)^2 r^2 dr 
+ [\frac{v^2 r^2}{2}]_{r=0}^{r=R}
\end{array}
\]
and integrate the representation
$v^2(R)R= v^2(\rho)R + \int_\rho ^R 2v\partial_r v Rdr$
to estimate the last term
\[
\begin{array}{ll}
v^2(R)\frac{R^2}{2}&=\int_{R/2}^R v^2(r)R dr + \int_{R/2}^R\int_\rho ^R 2v\partial_r v Rdr d\rho \\
&\ge \int_{R/2}^R v^2(r)R dr -  \int_{R/2}^R \Big( \int_\rho ^R (\partial_r v)^2 r^2 R^{-1}dr
\int_\rho^R v^2  R^3r^{-2}dr\Big)^{1/2}d\rho\\
&\ge \int_{R/2}^R v^2(r)R dr - \frac{\epsilon}{4}  \int_{R/2}^R (\partial_r v)^2 r^2 dr
-\frac{1}{4\epsilon} \int_{R/2}^R v^2   \frac{R^4}{r^4} r^{2}dr
\end{array}
\]
which shows that
\[
\begin{array}{ll}
-\int_0^R \frac{1}{|x^j-X^n|}v^2 r^2 dr 
&\ge - \epsilon  \int_{0}^R (\partial_r v)^2 r^2 dr
-\frac{5}{\epsilon} \int_{0}^R v^2  \frac{R^4}{r^4} r^{2}dr.
\end{array}
\]
Similar bounds for the other interaction terms in $V$ implies that the bilinear form is coercive
\[
\begin{array}{ll}
&\int_{\tset^{3(J+N)-1}} v \bar V_\tau v +\gamma v^2 \ dxdX_0\\
&\ge 
\int_{\tset^{3(J+N)-1}} \frac{1}{4}\sum_{j=1}^J |\partial_{x^j}v|^2  
+\frac{1}{4M}\sum_{n=1}^N |\partial_{X^n}v|^2  + v^2 \ dxdX_0 ,
\end{array}
\]
for $\gamma\ge 60(MN+J)$. Analogous estimates show that the bilinear form is also 
continuous, i.e. there is a constant $C$ such that
\[
\int_{\tset^{3(J+N)-1}} v \bar V_\tau w +\gamma vw \ dxdX_0\le C
\|v\|_{H^1(\tset^{3(J+N)-1})}  \|w\|_{H^1(\tset^{3(J+N)-1})}. 
\]
The combination of coercivity and continuity in $H^1(\tset^{3(J+N)-1})$ implies, by the theory of
compact operators, that the spectrum of $\bar V$ consists of eigenvalues
with orthogonal eigenvectors in $L^2(\tset^{3(J+N)-1})$, see \cite{evans}.

\begin{remark}[The Madelungen equation]\label{madelung}
An alternative to \eqref{psi_eq} is to instead include the coupling term $-\frac{1}{2M}
\sum_j \Delta_{X^j}\psi$
in the eikonal equation, which leads to the so called Madelungen equations \cite{madelung}.
Near the minima points, where 
$E-V_0(X)=0$, the perturbation $-\psi\cdot \frac{1}{2M}\sum_j \Delta_{X^j}\psi$
can be negative and then there is no real solution $\partial_X\theta$ to the corresponding eikonal equation.
To have a non real velocity  $\partial_X\theta$ is in our case not compatible with a classical limit
and therefore we avoid the Madelungen formulation.
\end{remark}

\end{document}